\documentclass[a4paper,12pt]{amsart}

\usepackage{amssymb}
\usepackage{amsmath}
\usepackage{amsfonts}
\usepackage{tikz}

\def\C{\mathbb{C}}

\def\Z{\mathbb{Z}}
\def\M{\mathcal{M}}

\def\diag{\mathrm{diag}}
\def\Id{\mathbf{I}}

\def\Hom{\mathrm{Hom}}

\def\Auth{\mathrm{Aut}}
\def\ep{\epsilon}

\def\d{\partial}

\newcommand{\mf}[1]{\mathfrak{#1}}

\renewcommand{\appendix}[1]{
    \addtocounter{section}{1}
    \setcounter{equation}{0}
    \renewcommand{\thesection}{\Alph{section}}
    \section*{Appendix \thesection\protect\indent #1}
    \addcontentsline{toc}{section}{Appendix \thesection\ \ \ #1}
}
\newcommand\encadremath[1]{\vbox{\hrule\hbox{\vrule\kern8pt
\vbox{\kern8pt \hbox{$\displaystyle #1$}\kern8pt}
\kern8pt\vrule}\hrule}}
\def\enca#1{\vbox{\hrule\hbox{
\vrule\kern8pt\vbox{\kern8pt \hbox{$\displaystyle #1$}
\kern8pt} \kern8pt\vrule}\hrule}}

\newcommand\figureframex[3]{
\begin{figure}[bth]
\hrule\hbox{\vrule\kern8pt
\vbox{\kern8pt \vbox{
\begin{center}
{\mbox{\epsfxsize=#1.truecm\epsfbox{#2}}}
\end{center}
\caption{#3}
}\kern8pt}
\kern8pt\vrule}\hrule
\end{figure}
}
\newcommand\figureframey[3]{
\begin{figure}[bth]
\hrule\hbox{\vrule\kern8pt
\vbox{\kern8pt \vbox{
\begin{center}
{\mbox{\epsfysize=#1.truecm\epsfbox{#2}}}
\end{center}
\caption{#3}
}\kern8pt}
\kern8pt\vrule}\hrule
\end{figure}
}

\newcommand{\beq}{\begin{equation}}
\newcommand{\eeq}{\end{equation}}
\newcommand{\bea}{\begin{eqnarray}}
\newcommand{\eea}{\end{eqnarray}}

\newcommand{\Res}{\mathop{\,\rm Res\,}}

\renewcommand{\l}{\lambda}
\newcommand{\om}{\omega}

\renewcommand{\d}{{{\partial}}}

\newcommand{\Pint}{{\int\kern -1.em -\kern-.25em}}

\renewcommand{\l}{\lambda}

\renewcommand{\l}{\lambda}

\renewcommand{\thesection}{\arabic{section}}

\makeatletter
\@addtoreset{equation}{section}
\makeatother

\newcommand{\br}[1]{\left( #1 \right) }
\newcommand{\ba}[1]{\left\langle #1 \right\rangle }

\newtheorem{theorem}{Theorem}[section]
\newtheorem{proposition}[theorem]{Proposition}
\newtheorem{corollary}[theorem]{Corollary}
\newtheorem{lemma}[theorem]{Lemma}

\theoremstyle{definition}
\newtheorem{definition}[theorem]{Definition}
\newtheorem{notation}[theorem]{Notation}

\theoremstyle{remark}
\newtheorem{remark}[theorem]{Remark}

\def\br{\begin{remark}\rm\small}
\def\er{\end{remark}}
\def\bt{\begin{theorem}}
\def\et{\end{theorem}}
\def\bd{\begin{definition}}
\def\ed{\end{definition}}
\def\bp{\begin{proposition}}
\def\ep{\end{proposition}}
\def\bl{\begin{lemma}}
\def\el{\end{lemma}}
\def\bc{\begin{corollary}}
\def\ec{\end{corollary}}
\def\beaq{\begin{eqnarray}}
\def\eeaq{\end{eqnarray}}

\def\CP1{\mathbb{C}\mathrm{P}^1}

\title[Givental formula and topological recursion]
{Identification of the Givental formula with the spectral curve topological recursion procedure}

\author{P.~Dunin-Barkowski}

\author{N.~Orantin}

\author{S.~Shadrin}

\author{L.~Spitz}

\address{P.~D.-B.: Korteweg-de~Vries Institute for Mathematics, University of Amsterdam, P.~O.~Box 94248, 1090 GE Amsterdam, The Netherlands and ITEP, Moscow, Russia}
\email{P.Dunin-Barkovskiy@uva.nl}

\address{N.~O.: CAMGSD, Departamento de Matem\'atica,
Instituto Superior T\'ecnico,
Av. Rovisco Pais,
1049-001 Lisboa, Portugal}
\email{norantin@math.ist.utl.pt}

\address{S.~S.: Korteweg-de~Vries Institute for Mathematics, University of Amsterdam, P.~O.~Box 94248, 1090 GE Amsterdam, The Netherlands}
\email{S.Shadrin@uva.nl}

\address{L.~S.: Korteweg-de~Vries Institute for Mathematics, University of Amsterdam, P.~O.~Box 94248, 1090 GE Amsterdam, The Netherlands}
\email{L.Spitz@uva.nl}

\begin{document}

\begin{abstract} We identify the Givental formula for the ancestor formal Gromov-Witten potential with a version of the topological recursion procedure for a collection of isolated local germs of the spectral curve. As an application we prove a conjecture of Norbury and Scott on the reconstruction of the stationary sector of the Gromov-Witten potential of $\CP1$ via a particular spectral curve.
\end{abstract}

\maketitle

\tableofcontents

\section{Introduction}

\subsection{Givental theory}
Givental theory~\cite{Giv01,Giv01b,Giv04} is one of the most important tools in the study of Gromov-Witten invariants of target varieties and general cohomological field theories that allows, in particular, to obtain explicit relations between the partition functions of different theories, reconstruct higher genera correlators from the genus~$0$ data, and establish general properties of semi-simple theories.

The core of the theory is Givental's formula that gives a formal Gromov-Witten potential associated to a calibrated semi-simple Frobenius structure. Teleman proves \cite{Tel07} that the formal Gromov-Witten potential associated to the calibrated Frobenius structure of a target variety with semi-simple quantum cohomology coincides with the actual Gromov-Witten potential in all genera.

Roughly speaking, to a calibrated Frobenius structure of dimension~$r$ with a chosen semi-simple point $t$ one can associate two $r\times r$ matrix series, $S_t(\zeta^{-1})$ and $R_t(\zeta)$, and $r\times r$ matrices $\Psi_t$ and $\Delta_t$ (the latter one is diagonal), such that for a certain quantization of these matrices we have the following formula for the corresponding Gromov-Witten potential\footnote{The formula as appears here is missing one term corresponding to $g=1, n=0$ that we completely ignore in this paper.}:
\begin{equation}
\hat S_t^{-1} \hat\Psi_t \hat R_t \hat \Delta_t Z_{\mathrm{KdV}}^{\otimes r},
\end{equation}
where by $Z_{\mathrm{KdV}}$ we denote the Kontsevich-Witten tau-function of the KdV hierarchy, that is, the function parametrizing the intersection indices of $\psi$-classes on the moduli space of curves.

\subsection{Topological recursion theory}
The theory developed by Eynard and the second named author (see \cite{EO,EO08}), is a procedure, called topological recursion, that takes the following objects as input. First, a particular Riemann surface, which is usually called the \emph{spectral curve}. Second, two functions $x$ and~$y$ on this surface, and, third, a choice of a bi-differential on this surface, which we will call the \emph{two-point function} (it has often been referred to as the Bergman kernel, but since this term has other uses as well, we refrain from using it in this paper). And, occasionally, a particular extra choice of a coordinate on an open part of the Riemann surface. The output of the topological recursion is a set of $n$-forms $\omega_{g,n}$, whose expansion in this additional coordinate generates interesting numbers.

In some cases these numbers are correlators of a matrix model (that was the original motivation for introducing the topological recursion; it is a natural generalization of the reconstruction procedure for the correlators of a certain class of matrix models, see, e.g. \cite{AMM06a}), in some other cases they appear to be related to Gromov-Witten theory and to various intersection numbers on the moduli space of curves.

Note that this topological recursion is unrelated to the topological recursion occurring in the theory of moduli spaces of curves. Throughout this paper topological recursion is always understood in the above ``matrix model''-related sense. We refer to the whole corresponding theory as the ``topological recursion theory''. We will also sometimes refer to this side of the story as ``the spectral curve side''.

One of the ways to think about the input data of the topological recursion theory is to say that the $(g,n)=(0,1)$ part of a partition function in some geometrically motivated theory determines the spectral curve; the $(g,n)=(0,2)$ part of a partition function determines the two-point function, and the rest of the correlators can be reconstructed from these two via topological recursion, in terms of a proper expansion of $\omega_{g,n}$ (see \cite{DMSS12}).

The topological recursion theory is often used to reproduce known partition functions, extracts from it some higher genus correlators which were up to now unreachable and gives new non-trivial relations for the correlators, see e.~g.~\cite{EO12}.

\subsection{Goals of the paper}


As we see, there is a lot of similarity in both theories (which was first noted by Alexandrov, Mironov and Morozov in \cite{AMM04,AMM06a,AMM06}). In both cases we have to start with a small amount of data fixed in genus zero, and in both cases the intersection indices of $\psi$-classes on the moduli space of curves are a kind of structure constants of the reconstruction procedure (in the case of Givental it is just a part of Givental's formula for the formal Gromov-Witten potential, and in the case of topological recursion it is recovered locally in an expansion near a simple critical point of the spectral curve, see~\cite{Eynintersection1bp}).

Moreover, in both cases we have an expansion of the correlators in terms of Feynman graphs, see~\cite{PSL} on the Givental side and~\cite{Espectralcurve, EO12, KO} on the spectral curve side. So, the natural question is whether we can precisely identify both theories in some setup.

On the Givental side we restrict ourselves to a part of the Givental formula, namely, $\hat R \hat \Delta Z_{\mathrm{KdV}}^{\otimes r}$ (this expression gives the so-called \textit{total ancestor potential}, written in normalized canonical basis). In some sense, it is the most important part of the Givental formula since it determines the underlying Frobenius structure, while the rest of the formula is a linear change of variables (action of the matrix $\hat \Psi$) and a change of calibration rather than of the Frobenius structure itself (action of the matrix series $\hat S^{-1}$). Note that for a cohomological field theory which does not have quadratic terms in the potential, the $S$-action becomes trivial when one takes the origin as the chosen point on the Frobenius manifold. For Gromov-Witten applications, where quadratic terms do appear, the $S$-action is nontrivial, but, together with $\Psi$-action, it amounts to a linear change of variables. This still allows for the correspondence below to be established, as long as one makes a specific choice of coordinates on the topological recursion side. We describe this in detail in the case of the particular example of $\CP1$, see below.

On the topological recursion side we consider a collection of local germs of a spectral curve at a finite number of points, with fixed expansions of the coordinate functions $x$ and $y$ and the two-point function near these points. The result of the topological recursion are local germs of $n$-forms $\omega_{g,n}$ defined on the products of the given germs of the curve, which we expand in a particular basis of forms that also depends on the expansions of the two-point function.

The resulting systems of correlators coincide for consistent choices of the input data in both theories. We prove this fact, essentially using the graphical interpretation of the formulas given in \cite{PSL, Espectralcurve}, and provide a dictionary to translate Givental data into local spectral curve data and vice versa.

Thus, we solve the problem about the mysterious relation between topological recursion and enumerative geometry. Namely, we almost fully complete the program proposed in the thesis of the second named author \cite{OrT} (see also \cite{Or08,DMSS12}) aiming at building a map between a problem of enumerative geometry and the topological recursion setup for some spectral curve. The only issue which is not adressed in our work is the definition of a global spectral manifold which will have to be adressed elsewhere in relation with mirror symmetry and Picard-Lefschetz theory.

As an application we prove the Norbury-Scott conjecture on the stationary sector of the Gromov-Witten invariants of $\CP1$~(\cite{NS}). Namely, we identify the ingredients of their formulas with the matrix series $R$ in the Givental formula for $\CP1$ and show that the expansion of $\omega_{g,n}$ that Norbury and Scott propose exactly reproduces the $\hat S^{-1}\hat \Psi$-part of the Givental formula for the Gromov-Witten potential of $\CP1$.

\subsection{Organization of the paper} The paper assumes some pre-know\-ledge of both Givental and topological recursion theories, and we refer to~\cite{EO08,LeeAx,BCOV} as possible sources. In Section 2 we recall the Givental theory, and present the Givental formula as a sum over graphs. In Section 3 we do the same for the topological recursion. In Section 4 we prove the theorem on identification of the two theories and provide a corresponding dictionary. In Section 5 we provide the computations showing that this identification works for the spectral curve propose by Norbury and Scott and the Gromov-Witten theory of $\CP1$.

\subsection{Acknowledgments} The authors are very grateful to A.~Alexandrov, B.~Dub\-rovin, B.~Eynard, M.~Mulase, P.~Norbury and D.~Zvonkine for useful discussions. N. O. would like to thank the Korteweg-de Vries Institute for its hospitality where part of this work as been carried out. The work of N. O. is partly founded by the Funda\c{c}\~ao para a Ci\^{e}ncia e a Tecnologia through the fellowship SFRH/BPD/70371/2010. S.S. and L.S. were supported by a Vidi grant of the Netherlands Organization for Scientific Research (NWO). P.DB. was supported by NWO free competition grant 613.001.021 and also partially supported by the Ministry of Education and Science of the
Russian Federation under contract 14.740.11.067, by RFBR
grant 12-01-00525, by joint grant 11-01-92612-Royal Society and by the Russian Presidents Grant of Support for the Scientific Schools NSh-3349.2012.2.




\section{Givental group action as a sum over graphs}\label{section:Givental}
In this section we review the Givental group action and we remind the reader how it can be used to write the partition function of an $N$-dimensional semi-simple cohomological field theory as an operator acting on the product of~$N$ KdV $\tau$-functions. Using this, we write the partition function for such a cohomological field theory as a sum over decorated graphs. This is essentially the same as what was done in~\cite{PSL}; in the present paper the contributions are distributed in a slightly different way over the components of the graph to make the comparison with the topological recursion.

\subsection{Givental group action}\label{sec:diff-oper}

We remind the reader of the original formulation, due to Y.-P.~Lee, of the infinitesimal Givental group action in terms of differential operators~\cite{Lee03, Lee08, Lee09}.

Consider the space of partition functions for $N$-dimensional cohomological field theories
\begin{equation}\label{eq:typeFPS}
Z = \exp\left(\sum_{g \geq 0} \hbar^{g-1} \mathcal{F}_g\right)
\end{equation}
in variables $v^{d,i}$, $d\geq 0$, $i=1,\dots,N$. There is a fixed scalar product $\eta_{ij}=\delta_{ij}$ on the vector space $V:=\langle e_1,\dots,e_N\rangle$ of primary fields corresponding to the indices $i=1,\dots,N$. Furthermore, we will denote by $e_{\bf{1}}$ the vector in $V$ that plays the role of the unit.

Later on we will also use the so-called \emph{correlators}
$$
\ba{\tau_{d_1}(e_{i_1})\tau_{d_2}(e_{i_2}) \cdots \tau_{d_k}(e_{i_k})}_g 
$$
which correspond to the coefficients of formal power series $\mathcal{F}_g$ in the following way:
\begin{equation}
\mathcal{F}_g = \sum \frac{\ba{\tau_{d_1}(e_{i_1})\tau_{d_2}(e_{i_2}) \cdots \tau_{d_k}(e_{i_k})}_g}{|\Auth((i_m,d_m)_{m=1}^k)|} t^{d_1, i_1} \cdots t^{d_k, i_k} ,
\end{equation}
where $|\Auth((i_m, d_m)_{m=1}^k)|$ denotes the number of automorphisms of the collection of multi-indices~$(i_m,d_m)$ and where the sum is such that it includes each monomial~$t^{d_1, i_1} \cdots t^{d_k, i_k}$ exactly once. Note that in the special case of a Gromov-Witten theory for some manifold $X$, these correlators carry the following meaning:
\begin{multline}
\ba{\tau_{d_1}(e_{i_1})\tau_{d_2}(e_{i_2}) \cdots \tau_{d_k}(e_{i_k})}_g =\\
\sum_{\mathrm{deg}}\int_{[X_{g,k,\mathrm{deg}}]}ev^*_1(e_{i_1})\psi^{d_1}_1\,ev^*_2(e_{i_2})\psi^{d_2}_2 \cdots ev^*_k(e_{i_k})\psi^{d_k}_k,
\end{multline}
where $[X_{g,k,\mathrm{deg}}]$ is the moduli space of degree $\mathrm{deg}$ stable maps to $X$ of genus-$g$ curves with $k$ marked points, $ev_i$ is the evaluation map at the $i^{\mathrm{th}}$ point and $\psi$ correspond to $\psi$-classes.

Consider a sequence of operators $r_l\in \Hom(V,V)$ for $l\geq 1$, such that the operators with odd (resp., even) indices are symmetric (resp., skew-symmetric). Then we denote by $(r_lz^l)\hat{\ }$ the following differential operator: \begin{align}
(r_lz^l)\hat{\ } := & -(r_l)_{\bf{1}}^i\frac{\d}{\d v^{l+1,i}}
+ \sum_{d=0}^\infty v^{d,i} (r_l)_i^j \frac{\d}{\d v^{d+l,j}}
 \label{eq:Al-quintized} \\
& +\frac{\hbar}{2} \sum_{m=0}^{l-1}(-1)^{m+1} (r_l)^{i,j}\frac{\d^2}{\d v^{m,i}\d
v^{l-1-m,j}}. \notag
\end{align}

Here the indices $i,j \in \{1, \ldots, N\}$ on~$r_l$ correspond to the basis $\{e_1, \ldots, e_N\}$ of~$V$, and the index~$\bf{1}$ corresponds to the unit vector~$e_{\bf{1}}$. When we write~$r_l$ with two upper-indices we mean as usual that we raise one of the indices using the scalar product~$\eta$.

Given such a sequence of operators~$r_l$, we define an operator series~$R(z)$ in the following way
\begin{equation}\label{eq:operatorSeries}
R(z) = \sum_{l= 0}^\infty R_l z^l := \exp\left(\sum_{l=1}^\infty r_l z^l\right) .
\end{equation}
The quantization $\hat{R}$ of this series is defined by
\begin{equation}\label{eq:quantizedOS}
\hat{R} = \exp\left(\sum_{l=1}^\infty \left((-1)^l r_l z^l\right)\hat{\ } \right) .
\end{equation}
Givental observed that the action of such operators~$\hat{R}$ on formal power series~$Z$ for which the number of $\psi$-classes (given by the first index of $v^{d,\mu}$) at any monomial of degree~$n$ is no more than $3g - 3 + n$, is well-defined. The main theorem of~\cite{FabShaZvo06} states that this action preserves the property that $Z$ is a generating function of the correlators of a cohomological field theory with target space~$(V,\eta)$ (see also~\cite{Kaz07,Tel07}).

\begin{remark}
Note that this definition of~$\hat{R}$ differs from the one in~\cite{PSL} by the sign~$(-1)^l$. It is needed here to agree with Givental's notation in Proposition~\ref{prop:Givental-trivial}, cf.~\cite[Proposition 7.3]{Giv01}. For the same reason, in order to agree with the conventions of Givental, we label in a matrix by the upper index the column and by the lower index the row.
\end{remark}

\subsection{Givental operator for a Frobenius manifold} \label{sec:Frobenius}
Let~$Z(\{t^{d,\mu}\})$ be the partition function of some $N$-dimensional semi-simple conformal cohomological field theory. We recall the construction (due to Givental~\cite{Giv01, Giv01b, Giv04}, see also Dubrovin~\cite{Dub98}) of an operator series~$R(z)$ as in the previous section whose quantization takes the product of $N$ KdV $\tau$-functions to~$Z$.

Let~$F$ be the restriction of~$\log(Z)$ to the genus zero part without descendents. Denote $t^{\mu} := t^{0,\mu}$. Then~$F$ can be interpreted as a formal Frobenius manifold with metric
\begin{equation}
\eta_{\alpha \beta} = \frac{\partial^3 F}{\partial t^{\mathbf{1}} \partial t^\alpha \partial t^\beta} 
\end{equation}
and Frobenius algebra structure~$c^\gamma_{\alpha \beta}$
\begin{equation}
c_{\alpha \beta \gamma}= \frac{\partial^3 F}{\partial t^\alpha \partial t^\beta \partial t^{\gamma} } .
\end{equation}

We can assume that $\eta_{\alpha\beta}=\delta_{\alpha+\beta,n+1}$ and $e_{\mathbf{1}}=e_1$. According to~\cite{Dub96} it is always possible by an appropriate choice of these flat coordinates $t^{\mu}$.



\subsubsection{Canonical coordinates} Another set of coordinates is given by the \emph{canonical coordinates} $\{u^i\}$ which can be found as solutions to Equation~(3.54) from~\cite{Dub96}, and have the property that $\{\partial_i:= \partial/\partial u^i\}$ forms a basis of canonical idempotents of the Frobenius algebra product. In these coordinates the metric is the identity matrix and the unit vector field is given by $e_\mathbf{1} = \d_1 + \cdots + \d_N$.

Define $\Delta_i := 1/(\d_i, \d_i)$ to be the inverse of the square of the length of the $i^{\mathrm{th}}$ canonical basis element, and call $\{\d/\d v^i := \Delta_i^{1/2} \d/\d u^i\}$ the normalized canonical basis in the tangent space. We denote the coordinates corresponding to this basis by $v^i$, and the formal variables corresponding to these coordinates by $v^{d,i}$. They are precisely the formal variables $v^{d,i}$ appearing in the previous section.

Let $U$ be the matrix of canonical coordinates $U = \diag(u^1, \ldots, u^N)$ and denote by $\Psi$ the transition matrix from the flat to the normalized canonical bases. That is, denoting $\mathrm{d}t = (\mathrm{d}t^1, \ldots, \mathrm{d}t^N)^{\mathrm{T}}$ and $\mathrm{d}u = (\mathrm{d}u^1, \ldots, \mathrm{d}u^N)^{\mathrm{T}}$, one has
\begin{equation}
\Delta^{-1/2}\mathrm{d}u = \Psi \mathrm{d}t ,
\end{equation}
where $\Delta = \diag(\Delta_1, \ldots, \Delta_N)$.

\begin{remark}
Note that $\Psi$ obtained with the help of the definition above depends on the point $p$ of the Frobenius manifold.
\end{remark}


\subsubsection{Recursion}\label{sub:recursion} Construct an operator series $R(z) = \sum_{k \geq 0} R_k z^k$ as in the previous section in the following way.

Recursively define the off-diagonal entries of~$R_k$ in normalized canonical coordinates by solving the equation
\begin{equation}
\Psi^{-1} \mathrm{d}  (\Psi R_{k-1}) = [\mathrm{d}U, R_k] .
\end{equation}
using~$R_0 = \Id$ as a base case. Construct the diagonal entries of $R_k$ by integrating the next equation
\begin{equation}
\Psi^{-1}  \mathrm{d} (\Psi R_{k}) = [\mathrm{d}U, R_{k+1}]
\end{equation}
using the fact that the diagonal entries of $[\mathrm{d}U,R_{k+1}]$ are equal to zero. To fix the integration constant, use the Euler equation
\begin{equation}
R_k = -(i_E \mathrm{d}R_k)/k ,
\end{equation}
where $E = \sum u^i \d_i$ is the Euler field (here we use the fact that we started with a conformal cohomological field theory).

This procedure recursively defines $R_k$ for all~$k$. The following proposition is essentially proved in Givental's papers~\cite{Giv01, Giv01b}.

\begin{proposition}\label{prop:Givental-trivial}
Let $F$ be a local $N$-dimensional Frobenius manifold structure, semisimple at the origin, and let $(R_k)$ be the series of operators constructed from this $F$ by the recursive procedure described above, at the origin. Let $\Psi$ and $\Delta$ be as above, taken at the origin as well. Then we have the following formula:
\begin{equation}\label{eq:Giv-F-0}
\mathcal{F}_0 = \mathop{\mathrm{Res}}_{\hbar=0} \mathrm{d}\hbar \cdot \log \hat{\Psi}\hat{R}\hat{\Delta} \mathcal{T} .
\end{equation}
Here $\mathcal{F}_0=\mathcal{F}_0(\{t^{d,\mu}\})$ is the genus $0$ descendant potential of cohomological field theory associated to $F$; $\mathcal{T}$ is the product of $N$ KdV tau-functions,
\[
\mathcal{T}:=Z_{\mathrm{KdV}}(\{u^{d,1}\})\cdots Z_{\mathrm{KdV}}(\{u^{d,N}\});
\]
$\hat{\Delta}$ replaces the variables of $i^{\mathrm{th}}$ KdV $\tau$-function according to $u^{d,i} = \Delta_i^{1/2} v^{d,i}$ 
and replaces $\hbar$ with~$\Delta_i \hbar$, while $\hat{\Psi}$ is the change of variables $v^{d,i} = \Psi^i_\nu t^{d,\nu}$. The unit for the $R$-action is given by $(\Psi^1_1,\dots,\Psi^N_1)$.
\end{proposition}

\begin{remark}\label{re:Teleman} In fact, using Teleman's result in~\cite{Tel07}, one has a refined version of Equation~\eqref{eq:Giv-F-0}:
\begin{equation}\label{Teleman}
Z = \hat{\Psi}\hat{R}\hat{\Delta} \mathcal{T} .
\end{equation}
Note that it holds for cohomological field theories. In the Gromov-Witten case, when quadratic terms in the potential cannot be neglected, there appears an additional complication, see the next remark below.
\end{remark}


\begin{remark}
\label{rem:GWcase}
Givental's formula~~\cite{Giv01} for a Gromov-Witten total descendant potential (without the $(g=1,n=0)$-term),
\begin{equation}
Z = \hat{S}^{-1}\hat{\Psi}\hat{R}\hat{\Delta} \mathcal{T},
\end{equation}
also includes the operator $\hat{S}$, given by
\begin{equation}
\label{eq:quantizedOS_S}
\hat{S} = \exp\left(\sum_{l=1}^\infty (s_l z^{-l})\hat{\ }\right),
\end{equation}
where the operators $ (s_l z^{-l})\hat{\ }$ are defined in the following way (see, e.~g., \cite[Section 4.2]{FabShaZvo06}):
\begin{align}
\label{eq:squant}
\sum_{l=1}^\infty (s_lz^{-l})\hat{\ } = &  -(s_1)_{1}^\mu \frac{\d}{\d t^{0,\mu}}
+ \frac1\hbar \sum_{d=0}^{\infty} (s_{d+2})_{1,\mu} \, t^{d,\mu}
\\ 
+ \sum_{
\substack{d=0\\ l=1}
}^\infty
(s_l)_\nu^\mu \, t^{d+l,\nu} \frac{\d}{\d t^{d,\mu}}
&
+ \frac1{2 \hbar} \sum_{
\substack{d_1,d_2 \\ \mu_2,\mu_2}
}
(-1)^{d_1} (s_{d_1+d_2+1})_{\mu_1,\mu_2} \, t^{d_1,\mu_1} t^{d_2,\mu_2}.\notag
\end{align}
Note that formula  \eqref{eq:quantizedOS_S} for the quantization of $S$ differs from the analogous formula \eqref{eq:quantizedOS} for $R$ by a factor of $(-1)^l$ in the exponent, which agrees with the definition in Givental's papers \cite{Giv01b,Giv01}.

The matrices $s_k$ are defined through the following relation:
\begin{equation}
S(z) =\sum_{k=0}^\infty S_k z^{-k}=  \exp\left(\sum_{l=0}^\infty s_l z^{-l}\right),
\end{equation}
where for $S(z)$, taken at a point $p$ of the Frobenius manifold, we have (see \cite{Giv01}), for any points $a$ and $b$ of the Frobenius manifold,
\begin{equation}
(a , b\, S_p):= (a,b)+\sum_{k=0}^{\infty} \left\langle \tau_0(a) \exp\left(\tau_0(p)\right)\tau_k(b)\right\rangle_0 z^{-1-k}.
\end{equation}
Here on the left hand side the brackets stand for the scalar product on the tangent space to the Frobenius manifold at $p$, and we used an identification of the tangent space with the whole Frobenius manifold, since in this case the Frobenius manifold is itself a vector space.
If $p$ is the origin, we have just
\begin{equation}
(a , b\, S):= (a,b)+\sum_{k=0}^{\infty} \left\langle \tau_0(a)\tau_k(b)\right\rangle \ z^{-1-k}.
\end{equation}

Note that this $S$ action is defined in the general case when the total descendant genus $0$ potential is known. For the case when only a Frobenius potential is specified, the choice of $S$ is then called a \emph{calibration} of the Frobenius manifold, see \cite{Giv01b, Dub98} for related details. In the case of cohomological field theory when we disregard quadratic terms, the $S$ action is trivial if $p$ is taken to be the origin.

It turns out that in most of the relevant cases, e.g. for the Gromov-Witten theory of $\CP1$ (see section \ref{s:gwcp1} below), the only relevant term in equation (\ref{eq:squant}) is
$$
\sum_{
\substack{d=0\\ l=1}
}^\infty
(s_l)_\nu^\mu \, t^{d+l,\nu} \frac{\d}{\d t^{d,\mu}},
$$
since $(s_1)_{1}^\mu$ vanishes and all other terms just change the unstable terms in the potential.

This means, that in these cases $\hat{S}^{-1}$ just performs a linear change of formal variables 
$t^{d,\mu}$ in the following way:
\begin{equation}
t^{d,\mu} \mapsto \sum_{m=d}^{\infty}(S_{m-d})^{\mu}_{\nu} t^{m,\nu}.
\end{equation}
\end{remark}

\subsection{Expressions in terms of graphs}\label{sec:expr-graphs}
In~\cite{PSL} the action of an operator series as in equation~\eqref{eq:quantizedOS} is written as a sum over graphs. By Remark~\ref{re:Teleman}, this allows us to construct the potential of any semi-simple conformal cohomological field theory as a sum over graphs. Here we repeat the construction of~\cite{PSL} in a slightly different way that will be more convenient for the comparison with the topological recursion formalism. Furthermore, we also include the action of~$\hat{\Delta}$. It is easy to see that the construction is equivalent to that of~\cite{PSL}.

\begin{notation}\label{not:Gamma}
Let~$\gamma$ be any graph. By a half-edge we mean either a leaf or an edge together with a choice of one of the two vertices it is attached to. By $V(\gamma)$, $E(\gamma)$, $H(\gamma)$ and~$L(\gamma)$ we denote the sets of vertices, edges, half-edges and leaves of~$\gamma$. For any vertex~$v$ of~$\gamma$, denote by~$H(v)$ the set of half-edges connected to~$v$.

Let~$\widetilde{\Gamma}$ be the set of all connected graphs~$\gamma$ together with a choice of disjoint splitting~$L(\gamma)= L^*(\gamma) \coprod L^{\bullet}(\gamma)$,  a labelling of the vertices by pairs $(g,i) \in \mathbb{Z}_{\geq 0} \times \{1, \ldots, N\}$ and a labelling of the elements of $H(\gamma)$ by non-negative integers, such that the label of a leaf in~$L^{\bullet}$ is always greater than one. The elements of~$L^*(\gamma)$ are called \emph{ordinary leaves}, the elements of~$L^{\bullet}$ are called \emph{dilaton leaves}. We denote by~$\Gamma$ the subset of all graphs in~$\widetilde{\Gamma}$ that are \emph{stable}; that is, any vertex labelled~$(0,i)$ for some~$i$ is of valence at least three.

For any graph~$\gamma$ denote by $\mathfrak{g}\colon V(\gamma) \rightarrow \Z_{\geq 0}$ and $\mathfrak{i} \colon V(\gamma) \rightarrow \{1, \ldots, N\}$ the maps that associate to any vertex its first and second label respectively, and by $\mathfrak{k} \colon H(\gamma) \rightarrow \Z_{\geq 0}$ the map that associates to any half-edge its label. Denote by $\mathfrak{v} \colon L(\gamma) \rightarrow V(\gamma)$ the map that associates to each leaf the corresponding vertex, and by $\mathfrak{v_1}, \mathfrak{v_2} \colon E(\gamma) \rightarrow V(\gamma)$ and by $\mathfrak{h}_1, \mf{h}_2\colon E(\gamma) \rightarrow H(\gamma)$ the maps that associate to an edge the first and second vertex, and the corresponding half-edges respectively.
\end{notation}

\begin{remark}
The labels introduced above are used to keep track of different data for the trivial cohomological field theory; $g$ is for the genus, $i$ for the primary field in canonical coordinates and the labeling of the marked half-edges is for the power of $\psi$-class.
\end{remark}

\begin{remark}
As in~\cite{PSL}, edges of a graph in~$\Gamma$ are considered to be oriented (this allows to define the maps~$\mathfrak{v}_1$ and~$\mathfrak{v}_2$ unambiguously); the final result does not depend on the orientation.
\end{remark}

Let $R(z)^i_j$ be the components of the operator series~$R(z)$ in normalized canonical basis as computed in Section~\ref{sub:recursion}. To each part of a graph~$\gamma\in \Gamma$ we assign some polynomial in formal variables $\hbar$ and~$v^{d,i}$. Here~$\hbar$ is used to keep track of the genus, while the first index of~$v^{d,i}$ keeps track of the number of $\psi$-classes and the second index keeps track of the normalized canonical coordinate.

\subsubsection{Leaves}
To each ordinary leaf~$l \in L^*$ marked by~$k$ attached to a vertex marked by the pair~$(g,i)$, we assign
\begin{equation}\label{eq:markedLeaf}
(\mathcal{L}^{*})^i_k(l) := [z^k] \left( \sum_{d\geq 0}\left( (R(-z))^i_j v^{d,j} z^d \right)  \right) ,
\end{equation}
which corresponds to the second term in~\eqref{eq:Al-quintized}.

To a dilaton leaf $\lambda \in L^{\bullet}(\gamma)$ marked by~$k$ attached to a vertex marked by~$(g,i)$ we assign
\begin{equation}\label{eq:unmarkedLeaf}
(\mathcal{L}^\bullet)^i_k(\lambda) := [z^{k-1}] \left(-(R(-z))^i_\mathbf{1} \right)  ,
\end{equation}
which corresponds to the first term in~\eqref{eq:Al-quintized}, which is called the \emph{dilaton shift}.

\subsubsection{Edges}
To an edge~$e$ connecting a vertex~$v_1$ marked by $(g_1,i_1)$ to a vertex~$v_2$ marked by~$(g_2,i_2)$ and with markings $k_1$ and~$k_2$ at the corresponding half-edges, we assign 
\begin{equation}
\mathcal{E}^{i_1,i_2}_{k_1,k_2}(e) := [z^{k_1} w^{k_2}] \left( \hbar \cdot \frac{\delta^{i_1 i_2} - \sum_s(R(-z))^{i_1}_s (R(-w))^{i_2}_s}{z+w}\right) .
\end{equation}
Note that this does not depend on the choice of ordering of the vertices and that it follows from the fact that $R(z)$ can be written as $R(z) = \exp(\sum r_l z^l)$ that the numerator on the right-hand side is equal to the product of~$(z+w)$ with some power series in $z$ and~$w$, so this definition makes sense.

\subsubsection{Vertices}
Let~$v$ be a vertex marked by~$(g,i)$ with~$n$ ordinary leaves and half-edges attached to it labelled by $k_1, \ldots, k_n$ and~$m$ more dilaton leaves labelled by $k_{n+1}, \ldots, k_{n+m}$. Then we assign to~$v$
\begin{equation}\label{eq:vertices}
\mathcal{V}^{(g,i)}_{\{k_1, \ldots, k_{n+m}\}}(v) := \hbar^{g-1} \int_{\bar{\mathcal{M}}_{g,n+m}} \psi_1^{k_1} \cdots \psi_{n+m}^{k_{n+m}} .
\end{equation}


\subsubsection{$Z$ as a sum over graphs}
It is easy to see that the sum over all graphs in~$\Gamma$ of the product of the contributions described above, weighted by the inverse order of the automorphism group of the graph, is equal to the graph-sum described in~\cite{PSL} (the only difference is that now we have specialized to the action on the trivial cohomological field theory, leading to $\psi$-class integrals~\eqref{eq:vertices} as vertex contributions).  Thus, we recover the partition function~$Z$ of the cohomological field theory we started with as a sum over~$\Gamma$:
\begin{multline}\label{eq:graphs-Givental}
(\hat{R}\hat{\Delta} \mathcal{T})(\{v^{d,j}\}) = \sum_{\gamma \in \Gamma} \frac{1}{|\Auth(\gamma)|} \\
\prod_{v \in V(\Gamma)} \hbar^{\mathfrak{g}(v)-1} (\Delta_{\mf{i}(v)})^{\frac{1}{2}(2\mf{g}(v) -2 +\mathrm{val}(v))} 
\left\langle \prod_{h \in H(v)} \tau_{\mathfrak{k}(h)}
\right\rangle_\mathfrak{g} \\
\prod_{e \in E(\gamma)} \mathcal{E}^{\mathfrak{i}(\mathfrak{v}_1(e)),\mathfrak{i}(\mathfrak{v}_2(e))}_{\mathfrak{k}(\mathfrak{h}_1(e)),\mathfrak{k}(\mathfrak{h}_2(e))}(e)
\prod_{l \in L^*(\gamma)} (\mathcal{L}^*)^{\mathfrak{i}(\mathfrak{v}(l))}_{\mathfrak{k}(l)}(l)
\prod_{\lambda \in L^{\bullet}(\gamma)} (\mathcal{L}^\bullet)^{\mathfrak{i}(\mathfrak{v}(l))}_{\mathfrak{k}(l)}(\lambda).
\end{multline}

\section{Topological recursion}\label{section:EO}

In this section, we define a local version of the topological recursion and write the corresponding invariants as a sum over graphs, which allows us to compare it to the Givental action in the next section.

\subsection{Local topological recursion}

We define a local version of the topological recursion in the following way. The term local refers to the fact that the data are all defined locally around the canonical coordinates without any reference to the possible existence of a global manifold where these functions can be defined.

\bd
For $N \in \mathbb{N}^*$, we call times a set of $N$ families of complex numbers $\left\{h_{k}^{i}\right\}_{k \in \mathbb{N}}$ for $i=1,\dots,N$ and jumps another set of $N \times N$ infinite families of complex numbers $\left\{B_{k,l}^{i,j}\right\}_{(k,l) \in \mathbb{N}^2}$ for $i,j = 1 , \dots, N$. We finally define a set of canonical coordinates $\left\{a_i\right\}_{i=1}^{N} \in \mathbb{C}^{N}$ subject to $a_i \neq a_j$ for $i \neq j$.

For all $i, j \in \{1, \ldots, N\}$, we define the following set of analytic functions and differential forms in a neighborhood of $0 \in \C$:
\beq
x^{i}(z):= z^2 + a_i \, , \quad y^{i}(z):= \sum_{k=0}^\infty h_k^{i} z^k
\eeq
and
\beq
{B}^{i,j}(z,z') = \delta_{i,j} {dz \otimes dz' \over \left(z-z'\right)^2} +  \sum_{k,l=0}^\infty B_{k,l}^{i,j} z^k z'^l dz  \otimes dz'.
\eeq

For $2g-2+n>0$, we define the genus $g$, $n$-point correlation functions $\om_{g,n}^{i_1,\dots,i_n}(z_1,\dots z_n)$  recursively by
\begin{multline}
\om_{g,n+1}^{i_0,i_1,\dots,i_n}(z_0,z_1,\dots, z_n) := \sum_{j=1}^N \Res_{z \to 0}
{\int_{-z}^{z} B^{i_0,j}(z_0,\cdot) \over 2 \left(y^{j}(z)-y^{j}(-z)\right) \mathrm{d}x^{j}(z)} \\
\left( \om_{g-1,n+2}^{j,j,i_1,\dots,i_n}(z,-z,z_1,\dots, z_n) + \phantom{\sum_{A\cup B = \{1,\dots, n\}} \sum_{h=0}^g 
} \right. \\ 
 \left. \sum_{A\cup B = \{1,\dots, n\}} \sum_{h=0}^g 
\om_{h,\left|A\right|+1}^{j,\mathbf{i}_A}(z,{\bf z}_A)
\om_{g-h,\left|B\right|+1}^{j,\mathbf{i}_B}(-z,{\bf z}_B)
\right) ,
\end{multline}
where for any set~$A$, we denote by $\mathbf{z}_A$ (resp., $\mathbf{i}_A$) the set~$\{z_k\}_{k \in A}$ (resp., $\{i_k\}_{k\in A}$), and where the base of the recursion is given by
\beq
\om_{0,1}^{i}(z) := 0;
\qquad
\om_{0,2}^{i,j}(z,z') := B^{i,j}(z,z') .
\eeq

For convenience, in the sequel we denote 
\begin{equation}
K^{i,j}(z,z') = { \int_{-z}^{z} B^{i,j}(z', \cdot) \over 2(y^j(z) - y^j(-z)) d x^j(z) } 
\end{equation}
and
\begin{equation}
\omega_{g,n}(\vec{z}) = \sum_{\vec{i}} \omega^{\vec{i}}_{g,n}(\vec{z}) \ , 
\end{equation}
where the length of $\vec{z}$ and~$\vec{i}$ is~$n$. 
\ed

\subsection{Correlation functions and intersection numbers}

The correlation functions built by this topological recursion can actually be written in terms of intersection of $\psi$ classes on the moduli space of Riemann surfaces. This result is a slight generalization of \cite{Eynintersection1bp,Espectralcurve} to the local topological recursion.

\subsection{One-branch point case}

The link between the topological recursion formalism and intersection numbers on the moduli space of Riemann surfaces comes from the application of this formalism to the Airy curve. This case corresponds to $N = 1$ and:
\beq
x(z) = z^2 +a\, , \quad y(z) = z \quad \hbox{and} \quad
{B}(z,z') = {dz \otimes dz' \over \left(z-z'\right)^2} .
\eeq

\begin{remark}
Since there is only one branch point in this case, i.e. $N=1$, we omit the superscript indicating which branch point we consider in the notations of this section.
\end{remark}

For further convenience, we introduce two additional parameters by considering the curve
\beq
x(z) = z^2 +a\, , \quad y(z) = \alpha z \quad \hbox{and} \quad
{B}(z,z') = \beta {dz \otimes dz' \over \left(z-z'\right)^2} ,
\eeq
the usual Airy curve being $\alpha = \beta = 1$. In this case, the topological recursion reads
\begin{multline}
\om_{g,n+1}(z_0,z_1,\dots, z_n) := \Res_{z \to 0}
{\beta \over 2 \alpha} {dz_0 \over 2 z \, dz} {1 \over (z_0^2 - z^2)}  \\
\left( \om_{g-1,n+2}(z,-z,z_1,\dots, z_n) + \phantom{\sum_{A\cup B = \{1,\dots, n\}} \sum_{h=0}^g 
} \right. \\ 
 \left. \sum_{A\cup B = \{1,\dots, n\}} \sum_{h=0}^g 
\om_{h,\left|A\right|+1}(z,{\bf z}_A)
\om_{g-h,\left|B\right|+1}(-z,{\bf z}_B)
\right)
\end{multline}
and one has
\bl
The correlation functions of the Airy curve can be expressed in terms of intersection numbers:
\begin{multline}
\om_{g,n}(z_1,\dots,z_n) = 
\\
\left(-{\beta \over 2 \alpha} \right)^{2g+n-2} \, \beta^{g+n-1} \sum_{\alpha_1,\dots,\alpha_n \geq 0} \left< \tau_{a_1} \dots \tau_{a_n}\right>_{g,n} \prod_{i=1}^n {(2 \alpha_i+1)!! \, dz_i \over z_i^{2 \alpha_i+2}}.
\end{multline}

\el
This lemma was proved many times by direct computation \cite{EynardMumford, Eynintersection1bp,EO08,Zhou}, matching the topological recursion with the recursive definition of the intersection numbers.

As a side note, the first few correlation functions are
\beq
\om_{0.3}(z_1,z_2,z_3) = -{\beta^3 \over 2 \alpha} \prod_{i=1}^3 {dz_i \over z_i^2},
\eeq
\beq
\om_{0,4}(z_1,z_2,z_3,z_4) = {\beta^5 \over 4 \alpha^2}  \prod_{i=1}^4 {dz_i \over z_i^2} \, \sum_{i=1}^4 {3 \over z_i^2},
\eeq
\beq
\om_{1,1}(z) = {-\beta^2 \over 2 \alpha} {dz \over 8 z^4}
\eeq
and
\beq
\om_{1,2}(z_1,z_2) = {\beta^4 \over 4 \alpha^2} {dz_1 \, dz_2 \over 8 z_1^2 z_2^2} \left({5 \over z_1^4}+ {5 \over z_2^4} + {3 \over z_1^2 z_2^2}\right).
\eeq

\br\label{ref:conventions}
It is important to remark that there exist different conventions in the literature for defining the topological recursion, mainly differing by a change of sign of the recursion kernel. The latter can be recovered by a change of sign $\alpha \to - \alpha$.

%

\er

Let us now consider a deformation of the Airy curve which we will refer to as the KdV curve in the following. It has only one branch point, $N=1$, and reads
\beq
\left\{ \begin{array}{l}
x(z) = z^2 + a_i \cr
y(z) = \alpha {\displaystyle \sum_{k=1}^\infty} h_k z^k \cr
B(z,z') = \beta B_{\mathrm{KdV}}(z,z') = \beta  {dz \otimes dz' \over (z-z')^2} \cr
\end{array} \right. .
\eeq
The corresponding correlation functions can also be expressed in terms intersection numbers as follows:
\bl
The correlation functions of the KdV curve read:
\begin{multline}
\om_{g,n}(z_1,\dots,z_n)= 
 \left(-{\beta \over 2 \alpha h_1} \right)^{2g+n-2} \, \beta^{g+n-1} \sum_{m=0}^\infty {\left(-1\right)^m \over m!} \\ \sum_{\vec{\alpha} \in \mathbb{N^*}^m} \prod_{k=1}^m (2 \alpha_k+1)!! {h_{2\alpha_k+1} \over h_1} 
\prod_{i=1}^n {(2 d_i+1)!! \, dz_i \over z_i^{2 d_i+2}}  
\left< \prod_{j=1}^n \tau_{d_j} \,   \prod_{k=1}^m \tau_{\alpha_k+1} \right>_{g,n+m} .
\end{multline}

\el

\begin{proof}
Once again the proof can be found in the literature \cite{EynardMumford, EO-WP, Eynintersection1bp}. However, let us study a graphical interpretation of  this result when considering an arbitrary convention for the topological recursion.
For $f(z)$ an analytic function around $z \to 0$ and $\left\{T_k\right\}_{k \in \mathbb{Z}}$ a set of parameters, one can compute
\begin{multline}
 \Res_{Z_1 \to 0} \Res_{Z_2 \to 0} K(Z_1,z) \left\{\left(\sum_{k\geq 1} T_k Z_1^k\right) dZ_1  \, K(Z_2,-Z_1) f(Z_2) \left[dZ_2\right]^2 \right. \\
 \left. - \left(\sum_{k\geq 1} T_k (-Z_1)^k\right) dZ_1 \, K(Z_2,Z_1) f(Z_2) \left[dZ_2\right]^2 \right\}
\end{multline}
where the recursion kernel is the one of the Airy curve, i.e. the one for which $h_k=0$ for $k \geq 2$:
\beq
K(z,z_0) = {\beta \over 2 \alpha h_1} {dz_0 \over 2 z \, dz} {1 \over (z_0^2 - z^2)}.
\eeq
One can move the integration contours to get
\beq
 \Res_{Z_1 \to 0} \Res_{Z_2 \to 0} =  \Res_{Z_2 \to 0} \Res_{Z_1 \to 0} +  \Res_{Z_2 \to 0} \Res_{Z_1 \to Z_2} +  \Res_{Z_2 \to 0} \Res_{Z_1 \to -Z_2} .
\eeq


The first term of the right hand side vanishes since the integrand does not have any pole at $Z_1 \to 0$. Let us now compute one of the other two terms:
\begin{multline}
 \Res_{Z_2 \to 0} \Res_{Z_1 \to Z_2} K(Z_1,z) \left(\sum_{k\geq 1} T_k Z_1^k dZ_1\right)  \, K(Z_2,-Z_1) f(Z_2) \left[dZ_2\right]^2 = \\
-  \Res_{Z_2 \to 0} {\beta \over 2 \alpha h_1}  {dZ_2  \over 2 Z_2 } f(Z_2) \Res_{Z_1 \to Z_2}   {dz \over 2 Z_1 } {1 \over (z^2 - Z_1^2)} {\beta \over 2 \alpha h_1} {1 \over (Z_1^2 - Z_2^2)}  \left(\sum_{k\geq 1} T_k Z_1^k dZ_1\right)  = \\
 =  - \Res_{Z_2 \to 0} {\beta \over 2 \alpha h_1}  {dZ_2  dz \over 2 Z_2 } f(Z_2) {1 \over (z^2 - Z_2^2)} {\beta \over 2 \alpha h_1} \sum_{k\geq 1} {T_k \over 4} Z_2^{k-2} .
\end{multline}
In the same way, 
\begin{multline}
 \Res_{Z_2 \to 0} \Res_{Z_1 \to - Z_2} K(Z_1,z) \left(\sum_{k\geq 1} T_k Z_1^k dZ_1\right)  \, K(Z_2,-Z_1) f(Z_2) \left[dZ_2\right]^2 = \\
 = - \Res_{Z_2 \to 0} {\beta \over 2 \alpha h_1}  {dZ_2  dz \over 2 Z_2 } f(Z_2) {1 \over (z^2 - Z_2^2)} {\beta \over 2 \alpha h_1} \sum_{k\geq 1} {T_k \over 4} (-Z_2)^{k-2} .
\end{multline}
The sum of these two terms reads
\begin{multline}
 \Res_{Z_2 \to 0} \Res_{Z_1 \to \pm Z_2} K(Z_1,z) \left(\sum_{k\geq 1} T_k Z_1^k dZ_1\right)  \, K(Z_2,-Z_1) f(Z_2) \left[dZ_2\right]^2 = \\
 = -  \Res_{Z_2 \to 0} {\beta \over 2 \alpha h_1}  {dZ_2  dz \over 2 Z_2 } f(Z_2) {1 \over (z^2 - Z_2^2)} {\beta \over 2 \alpha h_1} \sum_{k\geq 1} {T_{2k} \over 2} (Z_2)^{2k-2}
 \end{multline}
and finally:
\begin{multline}
 \Res_{Z_1 \to 0} \Res_{Z_2 \to 0} K(Z_1,z) \left\{\left(\sum_{k\geq 1} T_k Z_1^k\right) dZ_1  \, K(Z_2,-Z_1) f(Z_2) \left[dZ_2\right]^2 \right. \\
 \left. - \left(\sum_{k\geq 1} T_k (-Z_1)^k\right) dZ_1 \, K(Z_2,Z_1) f(Z_2) \left[dZ_2\right]^2 \right\} = \\
 = \Res_{Z_2 \to 0} {\beta \over 2 \alpha h_1}  {dZ_2  dz \over 2 Z_2 } f(Z_2) {1 \over (z^2 - Z_2^2)}\left(- {\beta \over 2 \alpha h_1}\right) \sum_{k\geq 1} {T_{2k} } (Z_2)^{2k-2}.
\end{multline}

On the other hand, plugging in the times $h_k$ amounts to computing similar quantities:
\begin{multline}
\Res_{z \to 0}
{\beta \over 2 \alpha h_1} {dz_0 \over 2 z \, dz} {1 \over (z_0^2 - z^2)} {1 \over \left(1 + {\displaystyle \sum_{k=1}^\infty} {h_{2k+1} \over h_1} z^{2k} \right)} f(z) \left[dz\right]^2 = \\
= \Res_{z \to 0}
{\beta \over 2 \alpha h_1} {dz_0 \over 2 z \, dz} {1 \over (z_0^2 - z^2)}  f(z) \left[dz\right]^2 (1 - {\displaystyle \sum_{k=1}^\infty} {h_{2k+1} \over h_1} z^{2k} +\left[ {\displaystyle \sum_{k=1}^\infty} {h_{2k+1} \over h_1} z^{2k}\right]^2 + \dots )
\end{multline}
The first term of this sum is the Airy recursion kernel. The second one is of the shape of the preceding one with $T_{2k+2} = { 2 \alpha h_{2k+1}\over h_1}$ for $k \geq 1$ so that:
\begin{multline}
- \Res_{z \to 0}
{\beta \over 2 \alpha} {dz_0 \over 2 z \, dz} {1 \over (z_0^2 - z^2)}  f(z) \left[dz\right]^2   {\displaystyle \sum_{k=1}^\infty} {h_{2k+1} \over h_1} z^{2k}  = \\
=  \Res_{Z_1 \to 0} \Res_{Z_2 \to 0} K(Z_1,z_0) \left\{g(Z_1) dZ_1  \, K(Z_2,-Z_1) f(Z_2) \left[dZ_2\right]^2 \right. \\
 \left. - g(-Z_1) dZ_1 \, K(Z_2,Z_1) f(Z_2) \left[dZ_2\right]^2 \right\}
\end{multline}
where
\beq
g(z) := \sum_{k \geq 1}  { 2 \alpha h_{2k+1} \over \beta h_1} z^{2k+2}.
\eeq
This same procedure can be applied to the other terms of the sum. The $k^{\mathrm{th}}$ order term can be written as a sequence of $k+1$ residues computed with the Airy recursion kernel with $g(z) dz$ on one of the outgoing legs. This computation shows that introducing non-vanishing times amounts to introducing a non-vanishing $\om_{0,1}(z):=g(z) dz$ in the topological recursion.

It is often useful to represent the topological recursion in a graphical form by representing the interaction kernel $K(z,z_0)$ by an edge oriented from $z_0$ towards a trivalent vertex labeled by $z$ and the function $\om_{0,2}(z_1,z_2)$ by a non-oriented edge (see \cite{EO} for more details about this set of  graphs). In this form, $\om_{g,n}(z_1,\dots,z_n)$ is a sum over trivalent graphs of genus $g$ with $n$ leaves labeled by the arguments $z_1,\dots,z_n$. The preceding computation shows that the correlation functions of the KdV curve can be obtained from the correlation functions of the Airy curve by introducing a set of new leaves, called dilation leaves, in the definition of the graphs used. A dilation leave decorated by a label k is weighted by
\beq
(2d-1)!! \, \Res_{z \to 0} g(z) {dz \over z^{2d+1}} = (2d-1)!! \,  { 2 \alpha h_{2d-1} \over \beta } .
\eeq

%
%
Plugging this expression into the formula for the Airy correlation functions proves the result.
\end{proof}

\subsubsection{General case} In this section we give a formula for the correlation function of the local topological recursion.

\bd
Let $\Gamma_{g,n}$ be the subset of~$\Gamma$ (see Notation~\ref{not:Gamma}) consisting of graphs of genus~$g'$ such that $g' + \sum_{v \in V(\Gamma)} \mf{g}(v) = g$ and with~$n$ ordinary leaves.
%
%
%
%
%
%

%
\ed

\bt \label{thm:EOgraphsum}
The correlation functions can be written as a sum over decorated graphs whose vertices are weighted by intersection of $\psi$-classes on $\overline{\M}_{g,n}$, edges by the jumps, ordinary leaves by primitives of~$B$ and dilaton leaves by the times. 

For $2-2g-n<0$, one has 
\begin{multline}\label{eq:graph-EO}
\om_{g,n}(\vec{z}) = \sum_{\gamma \in \Gamma_{g,n}}
 \prod_{v \in V(\gamma)} \left(- 2 h^{\mf{i}(v)}_1\right)^{\chi_{\mf{g}(v),\mathrm{val}(v)}} 
 \left\langle\prod_{h \in H(v)} \tau_{\mf{k}(h)}  \right\rangle_{\mf{g}(v),\mathrm{val}(v)} \\
\prod_{e \in E(\gamma)} \check{B}_{\mf{k}(\mf{h}_1(e)),\mf{k}(\mf{h}_2(e))}^{\mf{i}(\mf{v}_1(e)),\mf{i}(\mf{v}_2(e))} 
 \prod_{l \in L^*(\gamma)} \sum_{j=1}^N \mathrm{d}\xi_{\mf{k}(l)}^{\mf{i}(\mf{v}(l))}(z_j,j) \, \prod_{\lambda \in L^{\bullet}(\gamma)} \check{h}_{\mathfrak{k}(\lambda)}^{\mf{i}(\mf{v}(\lambda))} 
\end{multline}
with 
\beq
\check{h}_k^{i}:=  2(2k - 1)!!  h_{2k - 1}^{i}  ,
\eeq
\beq
d\xi_d^{i}(z_\alpha,j):=\Res_{z \to 0} {(2d+1)!! dz \over z^{2d+2}} \int^z B^{i,j}(z,z_\alpha),
\eeq
\beq
\check{B}_{d_1,d_2}^{i,j} := B_{2d_1,2 d_2}^{i,j} \; (2d_1-1)!! \; (2 d_2 -1)!! 
\eeq
and
\beq
\left\langle\prod_{i=1}^n \tau_{k_i} \right\rangle_{g,n}:= \int_{\overline{\M}_{g,n}} \psi_1^{k_1} \psi_2^{k_2} \dots \psi_n^{k_n}.
\eeq

\et

\begin{proof}
The proof is very similar to the one presented in \cite{Espectralcurve,KO}. However, we prefer to present a completely graphical proof so that the link with the next sections becomes clear.

We follow the proof of \cite{Espectralcurve}. From the definition, one can write the correlation functions as a sum over graphs with oriented and non-oriented arrows linking trivalent vertices resulting in the following expression:
\beq
\om_{g,n}^{\vec{i}}(\vec{z}) = \sum_{G \in \widehat{G}_{g,n}} \om(G)
\eeq
with $ \widehat{G}_{g,n}$ the set of genus $g$ trivalent graphs with one root and $n-1$ leaves labeled by the arguments $z_i$ and a skeleton tree of oriented edges pointing from the root towards the leaves weighted by
\bea
\om(G) &=& \vec {\displaystyle \prod_{v \in V(G)} } \Res_{Z_v \to 0} \; \prod_{e \in E_{\mathrm{oriented}}(G)} K^{i(v_1(e)),i(v_2(e))}(Z_{v_1(e)},Z_{v_2(e)}) \cr
&& \;  \prod_{e \in E_{\mathrm{unoriented}}(G)} B^{i(v_1(e)),i(v_2(e))}(Z_{v_1(e)},Z_{v_2(e)})\cr
\eea
where each leaf is considered as a one-valent vertex $v$ and one denotes $Z_v$ the variable $z_i$ associated to this leaf in the correlation function, $E_{\mathrm{oriented}}(G)$ is the set of oriented leaves of $G$ and $E_{\mathrm{unoriented}}(G)$ is the set of unoriented leaves of $G$ (see \cite{EO} for further details). The product of residues $\vec {\displaystyle \prod_{v \in V(G)} } \Res_{Z_v \to 0}$ is oriented following the arrows, i.e. one first computes the residue corresponding to the end of an arrow before the one associated to its root.

It is useful to remark, that, for any edge, oriented or not, one has two types of contributions. Indeed, the functions $B^{i,j}(z,z')$ have a singular part
\beq
B_{\mathrm{KdV}}^{i,j}(z,z'):=\delta_{i,j}  {dz \otimes dz' \over (z-z')^2}
\eeq
and a regular part
\beq
B_{\mathrm{reg}}^{i,j}(z,z') := \sum_{k,l=0}^\infty B_{k,l}^{i,j} z^k z'^l dz  \otimes dz'
\eeq
when $z \to z'$:
\beq
B^{i,j}(z,z') = B_{\mathrm{KdV}}^{i,j}(z,z') + B_{\mathrm{reg}}^{i,j}(z,z').
\eeq
In the same way, one has
\beq
K^{i,j}(z,z') = K_{\mathrm{KdV}}^{i,j}(z,z') + K_{\mathrm{reg}}^{i,j}(z,z').
\eeq

One can translate this by representing $B_{\mathrm{KdV}}^{i,j}(z,z')$ (resp.  $B_{\mathrm{reg}}^{i,j}(z,z')$) by dashed (resp. dotted) unoriented edges and $K_{\mathrm{KdV}}^{i,j}(z,z')$ (resp.  $K_{\mathrm{reg}}^{i,j}(z,z')$) by dashed (resp. dotted) oriented edges form $z$ to~$z'$. The preceding sum is thus transformed into a sum over graphs where the edges are dotted or dashed and weighted accordingly.

The dashed edges can be expressed in a slightly different way. Indeed, one has
\beq
B_{\mathrm{reg}}^{i,j}(z,z') = \Res_{z_1 \to z}  \Res_{z_2 \to z'} B_{\mathrm{KdV}}^{i,i}(z,z_1) \, \left[\int^{z_1} \int^{z_2}  B_{\mathrm{reg}}^{i,j}(z_1,z_2)\right] \, B_{\mathrm{KdV}}^{j,j}(z_2,z')
\eeq
and
\beq
K_{\mathrm{reg}}^{i,j}(z,z') = \Res_{z_1 \to z}  \Res_{z_2 \to \pm z'} B_{\mathrm{KdV}}^{i,i}(z,z_1) \, \left[\int^{z_1} \int^{z_2}  B_{\mathrm{reg}}^{i,j}(z_1,z_2)\right] \, K_{\mathrm{KdV}}^{j,j}(z_2,z')
\eeq
by a simple application of the Cauchy formula. 



Remember that such an edge comes with integration of its boundary variables, thus, one typically has to compute
\beq\label{eqcluster}
\Res_{z \to 0} \Res_{z' \to 0} g(z) K_{\mathrm{reg}}^{i,j}(z,z')  f(z')
\eeq
which reads
\beq
\Res_{z \to 0} \Res_{z_1 \to z}  \Res_{z' \to 0} \Res_{z_2 \to \pm z'} g(z) B_{\mathrm{KdV}}^{i,i}(z,z_1) \, \left[\int^{z_1} \int^{z_2}  B_{\mathrm{reg}}^{i,j}(z_1,z_2)\right] \, K_{\mathrm{KdV}}^{j,j}(z_2,z') f(z').
\eeq
One can move the integration contours around $0$ thanks to:
\beq
\Res_{z \to 0} \Res_{z_1 \to  z}  = \Res_{z_1 \to 0} \Res_{z \to 0}  - \Res_{z \to 0} \Res_{z_1 \to 0}
\eeq 
and
\beq
 \Res_{z' \to 0} \Res_{z_2 \to \pm z'} = \Res_{z_2 \to 0} \Res_{z' \to 0}  - \Res_{z' \to 0} \Res_{z_2 \to 0}.
 \eeq
Since, the integrant does not have any pole as $z_1 \to 0$ nor $z_2 \to 0$, this shows that \ref{eqcluster} is equal to
\beq
\Res_{z_1 \to 0} \Res_{z \to 0}  \Res_{z_2 \to 0} g(z) \,  B_{\mathrm{KdV}}^{i,i}(z,z_1) \, \left[\int^{z_1} \int^{z_2}  B_{\mathrm{reg}}^{i,j}(z_1,z_2)\right] \, \Res_{z' \to 0} \; K_{\mathrm{KdV}}^{j,j}(z_2,z')\, f(z') .
\eeq

In the same way, one gets that
\beq
\Res_{z \to 0} \Res_{z' \to 0} g(z) B_{\mathrm{reg}}^{i,j}(z,z')  f(z')
\eeq
is equal to 
\beq
\Res_{z_1 \to 0} \Res_{z \to 0}  \Res_{z_2 \to 0} g(z) \,  B_{\mathrm{KdV}}^{i,i}(z,z_1) \, \left[\int^{z_1} \int^{z_2}  B_{\mathrm{reg}}^{i,j}(z_1,z_2)\right] \, \Res_{z' \to 0} \; B_{\mathrm{KdV}}^{j,j}(z_2,z')\, f(z') .
\eeq
One can finally proceed in a similar way for re-expressing the weights of the root and the leaves by writing\footnote{Remark that, for the roots and leaves, in opposition to the inner edges, the functions are the full ones, not just the regular part.}
\beq
 \Res_{z' \to 0} {K}^{i,j}(z,z')  f(z')
 =  \Res_{z_2 \to 0} \, \left[ \int^{z_2}  {B}^{i,j}(z,z_2)\right] \, \Res_{z' \to 0} \; K_{\mathrm{KdV}}^{j,j}(z_2,z')\, f(z') 
\eeq
and
\beq
\Res_{z \to 0} g(z) {B}^{i,j}(z,z')   = 
\Res_{z_1 \to 0} \Res_{z \to 0}  g(z) \,  B_{\mathrm{KdV}}^{i,i}(z,z_1) \, \left[\int^{z_1} {B}^{i,j}(z_1,z')\right] .
\eeq

As a result, by applying this transformation to each dotted line, any graph is composed of a set of dotted subgraphs whose vertices have the same label separated by dashed lines. Since each subgraph with label $i$ also includes a root and leaves, it is a contribution to the correlation functions obtained for the case $N = 1$, times $h_k^{i}$ and vanishing jumps $B^{i,i}_{k,l} = 0$. In the sum over graphs, one can thus replace every sum over such sub-graphs by vertices of corresponding genus weighted by the correlation function for $N = 1$, which reads
\beq
\om_{g,n}^{\vec{i}}(\vec{z}) = \sum_{\gamma \in \Gamma_{g,n}} \Omega(\gamma)
\eeq
where
\bea
\Omega(\gamma) &=& \prod_{v \in V(\gamma)} \prod_{h \in H(v)} \Res_{Z_h \to 0} \om_{g(v),\mathrm{val}(v)}^{\mathrm{KdV},i(v)}\left(\left\{Z_h\right\}_{h \in H(v)} \right) \cr
&& \quad \prod_{e \in E(\gamma)} \int^{Z_{h_1(e)}} \int^{Z_{h_2(e)}} B_{\mathrm{reg}}^{i(v_1(h)),i(v_2(h))}(Z_{h_1(e)},Z_{h_2(e)})\cr
&& \qquad \prod_{h \in L^*(\gamma)} \int^{Z_h} {B}^{i,j}(Z_h,z_h)
\eea
where $\om_{g,n}^{\mathrm{KdV},i}(z_1,\dots,z_n)$ is the genus $g$, $n$-pointed correlation function obtained from the topological recursion in the case $N = 1$ and the initial data:
\beq
\left\{ \begin{array}{l}
x(z) = z^2 + a_i \cr
y(z) = {\displaystyle \sum_{k=1}^\infty} h_k^{i} z^k \cr
B(z,z') = B_{\mathrm{KdV}}(z,z') = {dz \otimes dz' \over (z-z')^2} \cr
\end{array} \right. .
\eeq
As explained in the preceding section, it can be expressed in terms of intersection numbers:
\begin{multline}
\om_{g,n}^{\mathrm{KdV},i}(z_1,\dots,z_n)= 
 \left(-{2 h_1^i} \right)^{2-2g-n}  \sum_{m=0}^\infty {\left(-1\right)^m \over m!} \\ \sum_{\vec{\alpha} \in \mathbb{N^*}^m} \prod_{k=1}^m (2 \alpha_k+1)!! {h_{2\alpha_k+1}^i \over h_1^i} 
\prod_{i=1}^n {(2 d_i+1)!! \, dz_i \over z_i^{2 d_i+2}}  \\
\left< \prod_{j=1}^n \tau_{d_j} \,   \prod_{k=1}^m \tau_{\alpha_k+1} \right>_{g,n+m} 
\end{multline}
which can be made more symmetric under the exchange of the ordinary and dilation leaves by writing
\begin{multline}
\om_{g,n}^{\mathrm{KdV},i}(z_1,\dots,z_n)= 
\sum_{m=0}^\infty  \left(-{2 h_1} \right)^{2-2g-n-m}   {1 \over m!} \\ \sum_{\vec{\alpha} \in \mathbb{N^*}^m} \prod_{k=1}^m (2 \alpha_k-1)!! {2 h_{2\alpha_k-1}^i } 
\prod_{i=1}^n {(2 d_i+1)!! \, dz_i \over z_i^{2 d_i+2}}  \\
\left< \prod_{j=1}^n \tau_{d_j} \,   \prod_{k=1}^m \tau_{\alpha_k} \right>_{g,n+m} 
\end{multline}

%

Absorbing the factors of the form
\begin{equation}
\frac{(2d+1)!! \mathrm{d}z}{z^{2d+2}}
\end{equation} 
into the corresponding half-edge contribution, the weight of an inner edge becomes
\beq
\Res_{z_1\to 0} \Res_{z_2 \to 0} \int^{z_1} \int^{z_2} B_{\mathrm{reg}}^{i,j}(z_1,z_2) \; {(2d_1+1)!! \mathrm{d}z_1 \over z_1^{2d_1+2}}  \; {(2d_2+1)!! \mathrm{d}z_1 \over z_2^{2d_2+2}} 
\eeq
which is equal to
\beq
\check{B}_{d_1,d_2}^{i,j} := B_{2d_1,2 d_2}^{i,j} \; (2d_1-1)!! \; (2 d_2 -1)!! 
\eeq
while the weight of the ordinary leaves becomes 
\beq
\mathrm{d}\xi_d^{i}(z_\alpha,j):=\Res_{z \to 0} {(2d+1)!! \mathrm{d}z \over z^{2d+2}} \int^z B^{i,j}(z,z_\alpha) , 
\eeq
where one considers both the singular and non-singular part of $B^{i,j}(z,z_\alpha)$. Collecting these contributions together proves the theorem.
\end{proof}

\subsection{Change of scales}

An important property of the correlation functions built in this way is their homogeneity property which reads
\beq
\forall \l \in \mathbb{C}\, , \;  \om_{g,n}(\vec z_N|x,\l y,B) = \l^{2 -2g-n} \om_{g,n}(\vec z_N|x, y,B)
\eeq

One can thus get an additional factor $\l^{i}$ by replacing $h_k^{i} \to \l^{i} h_k^{i}$ resulting in a rescaling of the weight of the vertices by $\left(\l^{\mf{i}(v)}\right)^{2-2g(v) - \mathrm{val}(v)}$.

\subsection{Weights, Laplace transform and recursive definition}

It is interesting to note that the weights of the edges are the coefficient of the Laplace transform of $B$:
\beq
\check{B}^{i,j}(u,v) := \sum_{(k,l)\in \mathbb{N}^2} \check{B}_{k,l}^{i,j} u^{-k} v^{-l}
\eeq
is equal to
\begin{multline}
\check{B}^{i,j}(u,v)  = \delta_{i,j} {uv \over u+v} + \\
{\sqrt{uv} e^{u a_i+ v a_j} \over 2 \pi} \int_{x(z)-a_i \in \mathbb{R}^+}
\int_{x(z')-a_j \in \mathbb{R}^+} B^{i,j}(z,z') e^{-ux(z)-vx(z')} .
\end{multline}
In~\cite{Espectralcurve}, it was proved that, if $dx$ is a meromorphic form defined on a Riemann surface, $\check{B}^{i,j}(u,v)$ can be factorized and expressed in terms of some basic functions. Here, we will consider the converse and build $B_{k,l}^{i,j}$ by induction in such a way that there exist a set of functions $\left\{f_{i,j}(u)\right\}_{i,j=1}^{N}$ such that
\beq \label{eq:B-f}
\check{B}^{i,j}(u,v) = {uv \over u+v} \left(\delta_{i,j}  - \sum_{k=1}^{N} f_{i,k}(u) f_{k,j}(v) \right).
\eeq

Let us define the coefficients $B_{k,l}^{i,j}$ recursively in terms of the initial data $B_{k,0}^{i,j}$ by imposing that
\beq \label{eq:xi-recursion}
\xi_{d+1}^{i}(z,j) := - {d  \xi_{d}^{i}(z,j) \over dx^{[j]}(z)} - \sum_{k=1}^{\mathfrak b} \check{B}_{d,0}^{i,k} \, \xi_{0}^{k}(z,j),
\eeq
or, in terms of the Laplace transform
\begin{align}
f_d^{i}(u,j) &:= {\sqrt{u} \over 2 \sqrt{\pi}} \int_{x(z)-a_j \in \mathbb{R}^+} e^{-u(x(z)-a_j)} dx^{j}(z) \, \xi_d^{i}(z) 
\\ \notag
&= \delta_{i,j} (-1)^d u^d - \sum_{d'} \check{B}_{d,d'}^{i,j} u^{-d'-1},
\end{align}
it reads
\beq
f_{d+1}^{i}(u,j) := - {u  f_{d}^{i}(u,j)} - \sum_{k=1}^{N} \check{B}_{d,0}^{i,k} \, f_{0}^{k}(z,j).
\eeq

With this definition, one has
\beq
\check{B}^{i,j}(u,v) = {uv \over u+v} \left(\delta_{i,j}  - \sum_{k=1}^{\mathfrak b} f_0^{k}(u,i) f_0^{k}(v,j) \right).
\eeq




\section{Identification of the two theories}

In this section we show how to find a local spectral curve corresponding to any semi-simple conformal Frobenius manifold.

Suppose some local spectral curve is given. For any $i \in \{1, \ldots, N\}$ and $k \in \Z_{\geq 0}$ define
$$
W^i_k := \sum_{j=1}^N d\left( \left(-\frac{1}{2 z^j}\frac{\partial}{\partial z^j}\right)^k \xi_0^i(z^j, j) \right).
$$

\begin{theorem}
Let~$R$ be some series of operators on an $N$-dimensional vector space~$V$ as in Section~\ref{section:Givental}. Let $Z = \hat{R} \hat\Delta\mathcal{T}$, where $\mathcal{T}$ is a product of~$N$ KdV $\tau$-functions, be the partition function of the corresponding semi-simple cohomological field theory.

Define a local spectral curve by the following data
\begin{equation}\label{lsc1}
\check{B}^{i,j}_{p,q} := [z^p w^q] \frac{\delta^{ij}-\sum_{s=1}^N R^i_s(-z) R(-w)^j_s}{z+w}
\end{equation}
and
\begin{align}\label{lsc2}
\check{h}_k^{i} &:=  [z^{k-1}]  \left(-R(-z))^i_\mathbf{1} \right)\\
h^i_1 &:= -\frac{1}{2\sqrt{\Delta^i}}.
\end{align}
Let~$\omega_{g,n}$ be the genus~$g$, $n$-pointed topological recursion invariant of this spectral curve and denote by
$$
\Omega(\{v^{d,i}\}) = \left. \left(\sum_{g,d} \omega_{g,d}\hbar^{g-1} \right)\right|_{W_d^i = v^{d,i}}
$$
their sum after a change of variables~$W^i_k \leftrightarrow v^{d,i}$. Then the partition function of the cohomological field theory and the topological recursion invariants agree in the following sense:
\begin{equation}
Z(\{v^{d,i}\}) = \exp\left(\Omega(\{v^{d,i}\}) \right).
\end{equation}
\end{theorem}

\begin{proof}
In Sections \ref{section:Givental} and~\ref{section:EO} we have given representations of~$Z$ and~$\omega_{g,n}$ as sums over the set~$\Gamma$. We prove the theorem by showing that the contribution of each individual graph to~$Z$ is equal to the contribution to~$\Omega$.

Let $\gamma \in \Gamma$ be some graph. Note that on both sides we assign the same weight to the vertices of~$\gamma$, namely to a vertex labelled~$(g,i)$ with~$n$ half-edges attached to it labelled~$d_1, \ldots, d_n$ we associate
\begin{equation}
(-2h_1^i)^{2 - 2g - n} \left\langle \tau_{d_1} \cdots \tau_{d_n} \right\rangle_{g, n} .
\end{equation}
Furthermore, by equation~\eqref{lsc1}, any edge in~$\gamma$ contributes the same to $Z$ and~$\Omega$.

Let~$l$ be an ordinary leaf of~$\gamma$ labelled by~$k$ attached to a vertex labelled by~$(g,i)$. We use induction on~$k$ to show that the contribution to~$Z$ is the same as the contribution to~$\Omega$.

The contribution of~$l$ to~$Z$ is given by
\begin{equation}
\mathcal{L}^i_k(l)=  [z^k] \left(\sum_d (R(-z))^i_j v^{d,j} z^d \right) = \sum_{d= 0}^k (-1)^{k-d} (R_{k-d})^i_j v^{d,j}.
\end{equation}
When~$k=0$, the contribution of~$l$ to~$\Omega$ is given by
\begin{equation}
\sum_j d\xi_0^{i}(z_j,j) = W_0^i.
\end{equation}
Since~$(R_0)^i_j = \delta^i_j$, the contribution to $Z$ and~$\Omega$ agree when~$k=0$.

Now suppose that they agree for some $k \in \Z_{\geq 0}$. That is, suppose that
\begin{equation}
\sum_j d\xi_k^{i}(z^{j},j) = \sum_{l=0}^k (-1)^{k-l}(R_{k-l})^i_s W_l^s .
\end{equation}
Then, using Equation~\eqref{eq:xi-recursion}, the contribution of the leaf to~$\Omega$ for the index $k+1$ is given by
\begin{multline}
\sum_j d\xi_{k+1}^{i}(z^{j},j) = \sum_j d\left( -\frac{\partial \xi_k^{i}(z^{j},j)}{\partial x^{j}} - \sum_{t=1}^N \check{B}_{k,0}^{i,t} \xi_0^{t}(z^{j},j) \right) \\
= \sum_j d \left( - \frac{1}{2 z^j} \frac{\partial}{\partial z^{j}} \xi_k^{i} (z^{j}, j) - \sum_{t=1}^n -(-1)^{k+1}(R_{k+1})^i_t \xi_0^{t}(z^{j},j) \right) \\
= \sum_{l= 0}^k (-1)^l (R_{l})^i_t W_{k+1-l}^t + (-1)^{k+1}(R_{k+1})^i_t W_0^t = \sum_{l = 0}^{k+1} (-1)^l(R_{l})^i_t W_{k+1-l}^t ,
\end{multline}
where we used equation~\eqref{lsc1} to write
\begin{equation}
\check{B}^{i,t}_{k,0} = - (-1)^{k+1}(R_{k+1})^i_t.
\end{equation}

This completes the induction, and since it is clear that the dilaton leaves contribute the same in both cases, it also completes the proof of the theorem.
\end{proof}

%
%
\begin{remark}
The theorem above deals with the potential of a cohomological field theory written in terms of formal variables $v^{d,i}$ corresponding to normalized canonical basis. In order to pass to flat coordinates one can change the variables in the following way:
\begin{equation}
v^{d,i} = \Psi_\mu^i t^{d,\mu}.
\end{equation}
On the spectral curve side it will correspond to changing the variables $W^i_k$ in the following way:
\begin{equation}
W^i_k = \Psi_\mu^i V^{\mu}_k.
\end{equation}

Thus, the theorem holds in the same form for the potential of cohomological field theory written in terms of formal variables $t^{d,\mu}$, only one should identify $t^{d,\mu}$ with $V^{\mu}_d$.
\end{remark}

\begin{remark}
Above we established the correspondence between cohomological field theories and symplectic invariants of spectral curves. However, as noted in Remark~\ref{rem:GWcase}, in the case of Gromov-Witten theories we cannot disregard quadratic terms. So, in the formula for the total descendent potential an additional operator $\hat{S}$ appears. In some cases, again see Remark~\ref{rem:GWcase}, it performs only a linear change of formal variables $t^{d,\mu}$ on which the potential depends. Thus, to establish the correspondence in this case, one has to change the variables $W^i_k$ in precisely the same way, and then identify the resulting variables with $t^{d,\mu}$, similar to the case of previous remark. Occasionally, the changes of variables preformed by $\hat\Psi$ and $\hat S^{-1}$ can be a re-expansion of $\omega_{g,n}$ in a new coordinate on the spectral curve.  We explain this procedure in detail for the case of $\CP1$ below in section \ref{sec:NSconjecture}.
\end{remark}

\begin{remark} The system of equations obtained via a Laplace transform from the equations of Givental for the $R$-matrix (that is, the so-called equations of deformed flat connection) is studied in detail in~\cite[Section 5]{Dub98}. This gives, in particular, a recipe to reconstruct the two-point function directly from the Frobenius structure bypassing the reconstruction of the $R$-matrix. This also explains why we call the critical values  $a_1,\dots, a_N$ of $x$ the canonical coordinates.
\end{remark}




\section{The Norbury-Scott conjecture}
\label{sec:NSconjecture}

In this section we recall and prove the Norbury-Scott conjecture on the stationary sector of the Gromov-Witten theory of $\CP1$.

\subsection{Gromov-Witten theory of $\CP1$}\label{sec:CP1-Givental}
\label{s:gwcp1} The Gromov-Witten theory of $\CP1$ is discussed from the geometric point of view in many sources, see e.~g.~\cite{OP01}. Givental proved in~\cite{Giv01} that his formula for the formal Gromov-Witten potential coincides with the geometric Gromov-Witten potential of $\CP1$, so we discuss it here only from the Givental point of view, ignoring the geometric background. The same computations one can find in~\cite{Song-Song,Song}.

The underlying structure of Frobenius manifold is determined by the following solution of the WDVV equation
\begin{equation}
\frac{1}{2}(t^1)^2t^2+e^{t^2},
\end{equation}
and the scalar product given by
\begin{equation}
\left(\begin{matrix}
0 & 1 \\
1 & 0
\end{matrix}
\right).
\end{equation}
All ingredients of the Givental formula depend on a particular choice of the point on the Frobenius manifold, and in this case we choose the point $(0,0)$ in the coordinates $(t_1,t_2)$.

We perform a direct computation following the recipe of Givental in~\cite{Giv01b}, see also Section~\ref{sec:Frobenius}. As a possible choice of the canonical coordinates, we use
\begin{align}
u^1&=t^1+2\exp(t^2/2); \\
u^2&=t^1-2\exp(t^2/2).
\end{align}
In particular, for $t^1=t^2=0$ we have $u^1=-u^2=2$. Then,
\begin{align}
\Delta_1^{-1}=\left\langle \frac{\d}{\d u^1}, \frac{\d}{\d u^1} \right\rangle & = \frac{\exp(-t^2/2)}{2}; \\
\Delta_2^{-1}=\left\langle \frac{\d}{\d u^2}, \frac{\d}{\d u^2} \right\rangle & = \frac{-\exp(-t^2/2)}{2},
\end{align}
so we can choose the square roots as
\begin{align}
\Delta_1^{-1/2} & = \frac{\exp(-t^2/4)}{\sqrt 2}; \\
\Delta_2^{-1/2} & = \frac{-i\exp(-t^2/4)}{\sqrt 2},
\end{align}
and for this choice we have the following matrix of transition from the basis given by $(\d/\d t^1, \d/\d t^2)$ to the normalized canonical basis:
\begin{equation}
\Psi=\left(\begin{matrix}
\frac{\exp(-t^2/4)}{\sqrt{2}} & \frac{-i\exp(-t^2/4)}{\sqrt{2}} \\
\frac{\exp(t^2/4)}{\sqrt{2}} & \frac{i\exp(t^2/4)}{\sqrt{2}}
\end{matrix}\right).
\end{equation}
It is the matrix $\Psi=\Psi_\alpha^i$, where $\alpha$ labels the rows and corresponds to the flat basis, while $i$ labels the columns and corresponds to the normalized canonical basis.

The recipe of reconstruction of the matrix $R$ from~\cite{Giv01b} gives at the origin the matrix $R(\zeta)=\sum_{k=0}^\infty R_k \zeta^{k}$, where
\begin{equation}
R_k=\frac{(2k-1)!!(2k-3)!!}{2^{4k} k!}\cdot
\left(\begin{matrix}
-1 & (-1)^{k+1} 2ki \\
2ki & (-1)^{k+1}
\end{matrix}\right)
\end{equation}

The $S$ matrix is given by the derivatives of the deformed flat coordinates, computed in~\cite[Example 3.7.9]{DZ05}
At the origin we have:
\begin{align}
S(\zeta^{-1}) = & \Id + \zeta^{-1} \cdot
\left(\begin{matrix}
0 & 0 \\
1 & 0
\end{matrix}\right)
\\ \notag
&+\sum_{k=1}^\infty
\frac{\zeta^{-2k}}{(k!)^2}
\left(\begin{matrix}
1-2k\left(\frac{1}{1}+\cdots \frac{1}{k}\right) & 0 \\
0 & 1
\end{matrix}\right)
\\ \notag
& +\sum_{k=1}^\infty
\frac{\zeta^{-2k-1}}{(k!)^2}
\left(\begin{matrix}
0 & -2\left(\frac{1}{1}+\cdots \frac{1}{k}\right) \\
\frac{1}{k+1} & 0
\end{matrix}\right).
\end{align}
(Note once again that we are using the convention that the matrices are acting on vector rows, opposite to the standard one).

The unit vector at the origin in the normalized canonical basis is equal to
\begin{equation}
e=\left(1,0\right)\cdot \left(\begin{matrix}
\frac{1}{\sqrt{2}} & \frac{-i}{\sqrt{2}} \\
\frac{1}{\sqrt{2}} & \frac{i}{\sqrt{2}}
\end{matrix}\right)
=\left( \frac{1}{\sqrt{2}}, \frac{-i}{\sqrt{2}} \right).
\end{equation}

Therefore, the dilaton leaves (cf. Equation~\eqref{eq:unmarkedLeaf}) in the Givental formula for $\CP1$ at the origin are
\begin{align}
(\mathcal{L}^\bullet)^1_{k+1} & = \frac{1}{\sqrt 2} \cdot \frac{(-1)^{k+1} \left((2k-1)!!\right)^2}{k! 2^{4k}}; \\
(\mathcal{L}^\bullet)^2_{k+1} & = \frac{i}{\sqrt 2} \cdot \frac{\left((2k-1)!!\right)^2}{k! 2^{4k}}
\end{align}
for $k\geq 0$.

\begin{proposition}
The Gromov-Witten potential of $\CP1$,
\begin{equation}
Z_{\CP1}(\hbar,\{t^{\ell,1},t^{\ell,2}\}_{\ell=0}^\infty),
\end{equation}
is obtained from $\hat R \hat \Delta Z_{\mathrm{KdV}}^{\otimes 2}$ (understood as a sum over graphs in the sense of Section~\ref{sec:expr-graphs} and wirtten down in the normalized canonical basis, that is, in the variables $v^{d,i}$, $d\geq 0$, $i=1,2$) via a linear change of variables given by
\begin{equation}
\sum_{m\geq k}^\infty \left(t^{m,1}, t^{m,2}\right) S_k \zeta^{m-k}= \sum_{\ell=0}^\infty \left(v^{\ell,1}, v^{\ell,2}\right) \zeta^{\ell}\cdot \Psi^{-1},
\end{equation}
and a correction of the unstable terms (that is, $(g,n)$-correlators with $2g-2+n\leq 0$).
\end{proposition}

\begin{proof} In order to get the Gromov-Witten potential of $\CP1$ as given by the Givental formula, we have to apply the $\hat\Psi$- and  $\hat S^{-1}$-action to the expression in terms of graphs discussed in   Section~\ref{sec:expr-graphs} that corresponds to  $\hat R \hat \Delta Z_{\mathrm{KdV}}^{\otimes 2}$. The $\hat\Psi$-action is just a linear change of variable by definition. The general $S$-action is discussed in~\cite[Section 4.2]{FabShaZvo06}. It is a combination of a shift of variables that vanishes in our case (indeed, $(1,0)S_1=(0,0)$), the linear change of variables that we have in the statement of Proposition, and a correction of unstable terms that is not essential for us.
\end{proof}

\subsection{The Norbury-Scott conjecture}

Norbury and Scott~\cite{NS} propose the following construction. They consider a spectral curve given by
\begin{equation}
\left\{
\begin{array}{rcl}
x & = & z + \frac{1}{z} ; \\
y & = & \log z,
\end{array}
\right.
\end{equation}
and the standard two-point function
\begin{equation}
B(z,z')=\frac{dz\otimes dz'}{(z-z')^2}.
\end{equation}

Via topological recursion they obtain the $n$-forms $\omega_{g,n}$ that they consider in the global variable $x$, and they conjecture the following theorem (they prove it for $g=0,1$):

\begin{theorem} \label{thm:N-S}For $2g-2+n>0$, we have:
\begin{equation}
\prod_{j=1}^{n}\left(-\Res_{x_j=\infty} \frac{1}{(a_j+1)!} x_1^{a_j+1}\right) \omega_{g,n}(x_1,\dots,x_n) = \langle \prod_{j=1}^n \tau_{2,a_j} \rangle_g,
\end{equation}
where $\langle \prod_{j=1}^n \tau_{2,a_j} \rangle_g$ is the corresponding correlator in $Z_{\CP1}$, that is, the coefficient of $\hbar^{g-1} \prod_{j=1}^n t_{2,a_j} / |Aut(a_1,\dots,a_n)|$ in $\log Z_{\CP1}$.
\end{theorem}

In the rest of this section we prove this theorem, identifying all ingredients of the topological recursion with the corresponding parts of the Givental formula.

\subsection{Proof of the Norbury-Scott conjecture}

\subsubsection{Local coordinates near the branch points}\label{sec:CP1-localcoord}

We denote the local coordinates by $z_1=\sqrt{x-2}$ and $z_2=\sqrt{x+2}$. Then we have:
\begin{align}
x  = z_1^2+2 &\mbox{ near}\ x=2,\ z=1,\ z_1=0; \\
x  = z_2^2-2 &\mbox{ near}\ x=-2,\ z=-1,\ z_2=0.
\end{align}
Therefore,
\begin{align}
z  & = 1+\frac{z_1^2}{2}\pm z_1\sqrt{1+\frac{z_1^2}{4}}; \\
z  & = -1+\frac{z_2^2}{2}\pm i z_2\sqrt{1-\frac{z_2^2}{4}}.
\end{align}
In both cases we choose $+$ for $\pm$.

\subsubsection{Expansion of $y$}

Recall that $y=\log z$. The direct computation shows:
\begin{align}
y  & = \int \frac{d z_1}{\sqrt{1+\frac{z_1^2}{4}}}; \\
y  & = \int \frac{-i\, d z_2}{\sqrt{1-\frac{z_2^2}{4}}};
\end{align}
Note that
\begin{align}
\frac{1}{\sqrt{1+\frac{z_1^2}{4}}} & = 1 + \sum_{k=1}^{\infty} z_1^{2k} \cdot \frac{(-1)^{k}(2k-1)!!}{k! 2^{3k}}; \\
\frac{-i}{\sqrt{1-\frac{z_2^2}{4}}} & = -i + \sum_{k=1}^{\infty} z_2^{2k} \cdot \frac{(-i)\cdot (2k-1)!!}{k! 2^{3k}}.
\end{align}
Therefore
\begin{align}
y & = z_1 + \sum_{k=1}^{\infty} z_1^{2k+1} \cdot \frac{(-1)^{k}(2k-1)!!}{k! 2^{3k}(2k+1)}; \\
y & = -i z_2+ \sum_{k=1}^{\infty} z_2^{2k+1} \cdot \frac{(-i)\cdot (2k-1)!!}{k! 2^{3k}(2k+1)}.
\end{align}

Thus the coefficients $\check h^i_{k+1}$, $k\geq 0$, are given by the following formulas:
\begin{align}
\check h^1_{k+1} & = 2\cdot \frac{(-1)^{k}\left((2k-1)!!\right)^2}{k! 2^{3k}}; \\
\check h^2_{k+1} & = 2 \cdot \frac{(-i)\cdot \left((2k-1)!!\right)^2}{k! 2^{3k}}.
\end{align}

\subsubsection{Matrix $f_{i,j}(w)$}\label{sec:CP1-fij} We use the following definition of the matrix $f_{ij}(w)$ (cf. Equation~\eqref{eq:B-f}):
\begin{equation}
f_{ij}(w)=\delta_{ij}-w \check B^{[ij]}(0,w^{-1}),
\end{equation}
where $w=v^{-1}$.  We use $\tilde B^{ij}_{0,l} = (B^{ij}_{\mathrm{reg}})_{0,2l} (2l-1)!!$, and the following expressions:
\begin{align}
B^{11}_{\mathrm{reg}}(0,z_1) & = \left[\frac{dz(z'_1)\otimes dz(z_1)}{(z(z'_1)-z(z_1))^2}
- \frac{dz'_1\otimes dz_1}{(z'_1-z_1)^2} \right]_{z'_1=0} \\
B^{12}_{\mathrm{reg}}(0,z_2) & = \left[\frac{dz(z'_1)\otimes dz(z_2)}{(z(z'_1)-z(z_2))^2}
\right]_{z'_1=0} \\
 B^{21}_{\mathrm{reg}}(0,z_1) & = \left[\frac{dz(z'_2)\otimes dz(z_1)}{(z(z'_2)-z(z_1))^2}
\right]_{z'_2=0} \\
 B^{22}_{\mathrm{reg}}(0,z_2) & = \left[\frac{dz(z'_2)\otimes dz(z_2)}{(z(z'_2)-z(z_2))^2}
- \frac{dz'_2\otimes dz_2}{(z'_2-z_2)^2} \right]_{z'_2=0}
\end{align}
Therefore,
\begin{align}
B^{11}_{\mathrm{reg}}(0,z_1) & = \frac{1}{z_1^2}\left(\frac{1}{\sqrt{1+\frac{z_1^2}{4}}}-1\right) \\
B^{12}_{\mathrm{reg}}(0,z_2) & = \frac{i}{4 (1-\frac{z_2^2}{4})^{3/2}} \\
B^{21}_{\mathrm{reg}}(0,z_1) & = \frac{i}{4 (1+\frac{z_1^2}{4})^{3/2}} \\
B^{22}_{\mathrm{reg}}(0,z_2) & = \frac{1}{z_2^2}\left(\frac{1}{\sqrt{1-\frac{z_2^2}{4}}}-1\right)
\end{align}
So, we have the following expansions:
\begin{align}
B^{11}_{\mathrm{reg}}(0,z_1) & = \sum_{k=0}^{\infty} z_1^{2k} \cdot \frac{(-1)^{k+1}(2k+1)!!}{(k+1)! 2^{3(k+1)}} \\
B^{12}_{\mathrm{reg}}(0,z_2) & = \sum_{k=0}^{\infty} z_2^{2k} \cdot \frac{i(2k+1)!!}{(k)! 2^{3k+2}} \\
B^{21}_{\mathrm{reg}}(0,z_1) & = \sum_{k=0}^{\infty} z_2^{2k} \cdot \frac{i(-1)^k(2k+1)!!}{(k)! 2^{3k+2}} \\
B^{22}_{\mathrm{reg}}(0,z_2) & = \sum_{k=0}^{\infty} z_2^{2k} \cdot \frac{(2k+1)!!}{(k+1)! 2^{3(k+1)}}.
\end{align}
The formulas for $f_{ij}(w)$ are then
\begin{align}
f_{11}(w) & = 1+\sum_{k=1}^{\infty} w^{k} \cdot \frac{(-1)^{k+1}(2k-1)!!(2k-3)!!}{k! 2^{3k}} \\
f_{12}(w) & = \sum_{k=1}^{\infty} w^{k} \cdot \frac{-i(2k-1)!!(2k-3)!!}{(k-1)! 2^{3k-1}} \\
f_{21}(w) & = \sum_{k=1}^{\infty} w^{k} \cdot \frac{(-1)^{k}i(2k-1)!!(2k-3)!!}{(k-1)! 2^{3k-1}} \\
f_{22}(w) & = 1+\sum_{k=1}^{\infty} w^{k} \cdot \frac{-(2k-1)!!(2k-3)!!}{k! 2^{3k}}
\end{align}
This coinsides with the formula for the $\sum_{k=0}^\infty R_k2^{k}(-w)^k$ at the point $(0,0)$.

\subsubsection{Comparison of the coefficient of $(g,n,m)$-vertex} In this section we consider a vertex of genus $g$ with $n$ attached half-edges or ordinary leaves, and $m$ dilaton leaves, with an associated intersection number $\langle \prod_{i=1}^n \tau_{d_i} \prod_{i=1}^m \tau_{a_i+1}\rangle_{g,n+m}$. 
There are vertices of type $1$ and type $2$, depending on the canonical coordinate that we associate to the vertex. We compare the coefficients that we associate to these vertices in the Givental case, using the data from Section~\ref{sec:CP1-Givental} in Formula~\eqref{eq:graphs-Givental}, and in the case of local topological recursion, using the data from Sections~\ref{sec:CP1-localcoord}-\ref{sec:CP1-fij} in Formula~\eqref{eq:graph-EO}.

The coefficients that we have in Formula~\eqref{eq:graph-EO} (at the vertex of the type $1$ and $2$ resp.):
\begin{equation}
(-2)^{2-2g-n-m} \qquad \mbox{and} \qquad  (2i)^{2-2g-n-m}.
\end{equation}

Let us compute how these coefficients change if we take into account all the differences between $R$-matrix and the dilaton leaves.
First, observe that the extra factor of $2^k$ in $R_k$ and, in addition, an extra factor of $\sqrt 2$ that we have to put by hand on each ordinary leave give us together the extra factors of
\begin{equation}
2^{\sum_{i=1}^n d_i}2^{n/2} \qquad \mbox{and} \qquad  2^{\sum_{i=1}^n d_i} 2^{n/2}.
\end{equation}
Then the quotient of the contributions of the dilaton leaves gives us the extra factors of
\begin{equation}
2^{\sum_{i=1}^m (a_i+1)}2^{m/2}(-1)^m\qquad \mbox{and} \qquad  2^{\sum_{i=1}^m (a_i+1)} 2^{m/2}(-1)^m.
\end{equation}
Let us assign by hand an extra factor of $(-1)^{2g-2+n}$ to each $(g,n,m)$-vertex. This way we get the following coefficients:
\begin{equation}
2^{g-1+n/2+m/2} \qquad \mbox{and} \qquad  2^{g-1+n/2+m/2}i^{2g-2+n+m}.
\end{equation}
These coefficients are precisely
\begin{equation}
\left((\Delta_1)^{1/2}\right)^{2g-2+n+m} \qquad \mbox{and} \qquad  \left((\Delta_2)^{1/2}\right)^{2g-2+n+m}.
\end{equation}

Therefore, the coefficient of $\prod_{k=1}^{\hat n} W_{d_k}^{i_k}$ in a graph of global genus $\hat g$ with $\hat n$ marked leaves in the Formula~\eqref{eq:graph-EO} for the set up of Norbury-Scott, multiplied by
\begin{equation}
2^{\hat n/2}(-1)^{2\hat g-2+\hat n}=(-\sqrt 2)^{\hat n},
\end{equation}
is equal to the coefficient of $\prod_{k=1}^{\hat n} t^{d_k,i_k}$ in the same graph in Formula~\eqref{eq:graphs-Givental}. This extra factor will be taken into account via rescaling of the variables by $-\sqrt 2$.

\subsubsection{The $\Psi$-action}

Let us apply the $\Psi$-operator to the leaves. After comparing the $R$-action with graph expansion given formulas~\eqref{eq:graphs-Givental} and~\eqref{eq:graph-EO}, and taking into account the extra factor of $-\sqrt 2$, we have the following identification of the marking on the leaves:
\begin{equation}
\sum_{a-b=c} (t^{a,1},t^{a,2}) S_b = \left( W^{1}_{c}, W^{2}_{c}\right) \Psi^{-1} / (-\sqrt 2).
\end{equation}
Here
\begin{align}
W^1_0 & =
\left. \frac{dz}{(1-z)^2}\right|_{z=z(z_1)}
+
\left. \frac{dz}{(1-z)^2}\right|_{z=z(z_2)} \\
W^2_0 & =
\left. \frac{i\, dz}{(1+z)^2}\right|_{z=z(z_1)}
+
\left. \frac{i\, dz}{(1+z)^2}\right|_{z=z(z_2)},
\end{align}
and
\begin{equation}
W^{i}_c= d\left( \left(-\frac{d}{dx}\right)^{c} \int W^i_0 \right),
\end{equation}
so we can work in the global coordinate $z$ rather than in the local coordinates $z_1,z_2$.

Since
\begin{equation}
\Psi^{-1}/(-\sqrt 2)=
\left(\begin{matrix}
\frac{-1}{2} & \frac{-1}{2} \\
\frac{-i}{2} & \frac{i}{2}
\end{matrix}\right),
\end{equation}
we have:
\begin{equation} \label{eq:S-change}
\sum_{a-b=c} (t^{a,1},t^{a,2}) S_b = \left(U^1_c, U^2_c\right),
\end{equation}
where
\begin{align}
U^1_0 & = \frac{1}{2} \left( -\frac{dz}{(1-z)^2} + \frac{dz}{(1+z)^2}\right)\\
U^2_0 & = \frac{-1}{ 2} \left( \frac{dz}{(1-z)^2} + \frac{dz}{(1+z)^2}\right)
\end{align}
and
\begin{equation}
U^i_c= d\left( \left(-\frac{d}{dx}\right)^{c} \int U^i_0 \right), \qquad i=1,2; c=0,1,2,\dots.
\end{equation}

\subsubsection{The $S$-action} The $S$-action is just a linear change of variables prescribed by Equation~\eqref{eq:S-change}. This means that we replace each $U^i_c$ with a linear combination of times $t^{a,j}$, $a\geq c$, where the coefficient of $t^{a,2}$ (this is the series of variables corresponding to the stationary sector) is equal to
\begin{equation}\label{eq:S-coeff-1}
\begin{cases}
0, & \mbox{if } a-c \mbox{ is even;} \\
\frac{1}{(k+1)\cdot (k!)^2}, & \mbox{if } a-c=2k+1.
\end{cases}
\end{equation}
for $i=1$, and 
\begin{equation}\label{eq:S-coeff-2}
\begin{cases}
\frac{1}{(k!)^2}, & \mbox{if } a-c=2k; \\
0, & \mbox{if } a-c \mbox{ is odd.}
\end{cases}
\end{equation}
for $i=2$.

Norbury and Scott make the same kind of a linear change of variables, with the coefficient of $t^{a,2}$ in $U^j_c$, $j=1,2$, given by
\begin{equation}\label{eq:S-coeff-NS}
-\Res_{x=\infty} \frac{1}{(a+1)!} x^{a+1}U^j_c = \frac{1}{(a+1)!} \Res_{z=0} \left(z+\frac{1}{z}\right)^{a+1} U^j_c.
\end{equation}

In order to complete the proof of Theorem~\ref{thm:N-S}, we have to check two things: (1) that the Norbury-Scott formula for the contribution depends only on the difference $a-c$; (2) that for $c=0$ Equation~\eqref{eq:S-coeff-NS} gives exactly the same coefficients as we have in Equations~\eqref{eq:S-coeff-1} and~\eqref{eq:S-coeff-2}.

The first thing follows directly from the formula. Indeed,
\begin{align}
-\oint \frac{x^{a+1}}{(a+1)!} d\left( \left(-\frac{d}{dx}\right)^{c} \int U^j_0 \right) &
= \oint \frac{x^{a}}{(a)!} \left( \left(-\frac{d}{dx}\right)^{c} \int U^j_0 \right) \, dx
\\ \notag
& = \oint \frac{x^{a-c}}{(a-c)!} \left( \int U^j_0 \right) \, dx
\\ \notag
& = - \oint \frac{x^{a+1-c}}{(a+1-c)!} \, d \left( \int U^j_0 \right).
\end{align}
In particular, we see that the coefficient is equal to $0$ if $a<c$.

Then, a direct computation shows that
\begin{align}
& \frac{1}{(a+1)!} \Res_{z=0} \left(z+\frac{1}{z}\right)^{a+1} U^1_0
\\ \notag
& = \frac{1}{(a+1)!} \Res_{z=0} \left(z+\frac{1}{z}\right)^{a+1} \frac{1}{2} \left( -\frac{dz}{(1-z)^2} + \frac{dz}{(1+z)^2}\right)
\\ \notag
& =  \frac{1}{(a+1)!} \Res_{z=0} \left(z+\frac{1}{z}\right)^{a+1} \frac{-2z\,  dz}{(1-z^2)^2}
\\ \notag
& =
\begin{cases}
0, & \mbox{if } a \mbox{ is even;} \\
\frac{-2}{(2k+2)!}\left(\binom{2k+2}{0}\cdot (k+1) + \binom{2k+2}{1}\cdot k + \cdots \binom{2k+2}{k}\cdot 1  \right) & \mbox{if } a=2k+1.
\end{cases}
\\ \notag
& =
\begin{cases}
0, & \mbox{if } a \mbox{ is even;} \\
\frac{-1}{(k+1)(k!)^2} & \mbox{if } a=2k+1.
\end{cases}
\end{align}
and
\begin{align}
& \frac{1}{(a+1)!} \Res_{z=0} \left(z+\frac{1}{z}\right)^{a+1} U^2_0
\\ \notag
& = \frac{1}{(a+1)!} \Res_{z=0} \left(z+\frac{1}{z}\right)^{a+1} \frac{-1}{2} \left( \frac{dz}{(1-z)^2} + \frac{dz}{(1+z)^2}\right)
\\ \notag
& =  \frac{-1}{(a+1)!} \Res_{z=0} \left(z+\frac{1}{z}\right)^{a+2} \frac{z\,  dz}{(1-z^2)^2}
\\ \notag
& =
\begin{cases}
\frac{-1}{(2k+1)!}\left(\binom{2k+2}{0}\cdot (k+1) + \binom{2k+2}{1}\cdot k + \cdots \binom{2k+2}{k}\cdot 1  \right) & \mbox{if } a=2k; \\
0, & \mbox{if } a \mbox{ is odd}
\end{cases}
\\ \notag
& =
\begin{cases}
\frac{-1}{(k!)^2} & \mbox{if } a=2k; \\
0, & \mbox{if } a \mbox{ is odd.}
\end{cases}
\end{align}

We see that there is an extra factor of $(-1)$ in all coefficients. This means that the $(g,n)$-correlation functions of Norbury-Scott differ from the stationary Gromov-Witten invariants of $\CP1$ by the factor of $(-1)^n$. But this factor is exactly the difference we must have because Norbury and Scott using a different convention on the sign in the topological recursion, cf.~Remark~\ref{ref:conventions}. This completes the proof of Theorem~\ref{thm:N-S}.

\end{document}